\documentclass[12pt]{llncs}
\usepackage{yhmath}
\usepackage{url}
\usepackage{savesym}

\usepackage{amsmath}
\usepackage{graphicx}
\usepackage{epstopdf}
\usepackage{wrapfig}
\usepackage{amssymb}
\usepackage{wasysym}

\usepackage{yhmath}
\usepackage{color}
\usepackage{multirow}
\usepackage[margin=1.2in]{geometry}

\usepackage[usenames,dvipsnames]{xcolor,colortbl}

\newcommand{\etal}{\emph{et~al.}}
\newcommand{\ie}{\emph{i.e.}}
\newcommand{\eg}{\emph{e.g.}}

\graphicspath{{images/}}

\urldef{\mailsa}\path||
\urldef{\mailsb}\path||
\urldef{\mailsc}\path||
\newcommand{\keywords}[1]{\par\addvspace\baselineskip
\noindent\keywordname\enspace\ignorespaces#1}

\renewenvironment{proof}{\emph{Proof.}}{\hfill $\Box$ \medskip\\}
\newtheorem{observation}{Observation}

\mainmatter  
\title{A Simple, Faster Method for Kinetic Proximity Problems\footnote{This work was partially supported by a University of Victoria Graduate Fellowship and by NSERC discovery grants.\\Preliminary versions of parts of this paper appeared in \textit{Proceedings of the 29th ACM Symposium on Computational Geometry} (SoCG 2013)~\cite{socg17-rahmati} and \textit{Proceedings of the 13th Scandinavian Symposium and Workshops on Algorithm Theory} (SWAT 2012)~\cite{DBLP:conf/swat/RahmatiZ12}.}}

\author{\small
Zahed Rahmati\footnote{Dept. of Computer Science, University of Victoria, Victoria, BC, Canada. Email: {\tt rahmati@uvic.ca}}
\and
Mohammad Ali Abam\footnote{Dept. of Computer Engineering, Sharif University of Technology, Iran. Email: {\tt abam@sharif.edu}}
\and \\
Valerie King\footnote{Dept. of Computer Science, University of Victoria, Victoria, BC, Canada. Email: {\tt val@uvic.ca}}
\and
Sue Whitesides\footnote{Dept. of Computer Science, University of Victoria, Victoria, BC, Canada. Email: {\tt sue@uvic.ca}}
\and
Alireza Zarei\footnote{Dept. of Mathematical Science, Sharif University of Technology, Tehran, Iran. Email: {\tt zarei@sharif.edu}}
}

\institute{\small
}

\toctitle{Lecture Notes in Computer Science}
\tocauthor{Authors' Instructions}

\titlerunning{Kinetic Proximity Problems}
\authorrunning{Z. Rahmati, M. A. Abam, and V. King, S. Whitesides, and Alireza Zarei}

\begin{document}

\maketitle

\begin{abstract}
For a set of $n$ points in the plane, this paper presents simple kinetic data structures (KDS's) for solutions to some fundamental proximity problems, namely, the \textit{all nearest neighbors} problem, the \textit{closest pair} problem, and the \textit{Euclidean minimum spanning tree} (EMST) problem. Also, the paper introduces KDS's for maintenance of two well-studied sparse proximity graphs, the \textit{Yao graph} and the \textit{Semi-Yao graph}.

\vspace{+5pt}
We use sparse graph representations, the \textit{Pie Delaunay graph} and the \textit{Equilateral Delaunay graph}, to provide new solutions for the proximity problems. Then we design KDS's that efficiently maintain these sparse graphs on a set of $n$ moving points, where the trajectory of each point is assumed to be a polynomial function of constant maximum degree $s$. We use the kinetic Pie Delaunay graph and the kinetic Equilateral Delaunay graph to create KDS's for maintenance of the Yao graph, the Semi-Yao graph, all the nearest neighbors, the closest pair, and the EMST. Our KDS's use $O(n)$ space and $O(n\log n)$ preprocessing time.

\vspace{+5pt}
We provide the first KDS's for maintenance of the Semi-Yao graph and the Yao graph. Our KDS processes $O(n^2\beta_{2s+2}(n))$ (resp. $O(n^3\beta_{2s+2}^2(n)\log n)$) events to maintain the Semi-Yao graph (resp. the Yao graph); each event can be processed in time $O(\log n)$ in an amortized sense.  Here, $\beta_s(n)={\lambda_s(n)/ n}$ is an extremely slow-growing function and $\lambda_s(n)$ is the maximum length of Davenport-Schinzel sequences of order $s$ on $n$ symbols.

\vspace{+5pt}
Our KDS for maintenance of all the nearest neighbors and the closest pair processes $O(n^2\beta^2_{2s+2}(n)\log n)$ events. For maintenance of the EMST, our KDS processes $O(n^3\beta_{2s+2}^2(n)\log n)$ events. For all three of these problems, each event can be handled in time $O(\log n)$ in an amortized sense. 

\vspace{+5pt}
Our \textit{deterministic} kinetic approach for maintenance of all the nearest neighbors improves by an $O(\log^2 n)$ factor the previous \textit{randomized} kinetic algorithm by Agarwal, Kaplan, and Sharir. Furthermore, our KDS is simpler than their KDS, as we reduce the problem to one-dimensional range searching, as opposed to using two-dimensional range searching as in their KDS.

\vspace{+5pt}
For maintenance of the EMST, our KDS improves the previous KDS by Rahmati and Zarei by a near-linear factor in the number of events.

\end{abstract}

\newpage
\section{Introduction}
The goal of the \textit{kinetic data structure framework}, which was first introduced by Basch, Guibas and Hershberger~\cite{basch_data_1999}, is to provide a set of data structures and algorithms that maintain attributes (properties) of points as they move. At essentially any moment, one may seek efficient answers to certain queries (\eg, what is the closest pair?) about these moving points. Taken together, such a set of data structures and algorithms is called a \textit{kinetic data structure} (KDS). Kinetic versions of many geometry problems have been studied extensively over the past 15 years, \eg, kinetic Delaunay triangulation~\cite{Albers_voronoidiagrams,DBLP:journals/dcg/Rubin13}, kinetic point-set embeddability~\cite{DBLP:conf/gd/RahmatiZ12}, kinetic Euclidean minimum spanning tree~\cite{DBLP:conf/iwoca/RahmatiZ11,basch_data_1999}, kinetic closest pair~\cite{Agarwal:2008:KDD:1435375.1435379,basch_data_1999}, kinetic convex hull~\cite{basch_data_1999,Alexandron:2007:KDD:1219156.1219201}, kinetic spanners~\cite{Abam:2010:SEK:1630166.1630284,Karavelas:2001:SKG:365411.365441}, and kinetic range searching~\cite{Agarwal:2003:IMP:846156.846166}.

Let $P$ be a set of $n$ points in the plane, and denote the position of each point $p$ by $p=(p_x,p_y)$ in a Cartesian coordinate system. In the kinetic setting, we assume the points are moving continuously with known trajectories, which may be changed to new known trajectories at any time. Thus the point set $P$ will sometimes be denoted $P(t)$, and an element $p=(p_x,p_y)$ by $p(t)=(p_x(t),p_y(t))$. For ease of notation, we denote the coordinate functions of a point $p_i(t)$ by $x_i(t)$ and $y_i(t)$. Throughout the paper we assume that all coordinate functions are polynomial functions of maximum degree bounded by some constant $s$.

In this paper, we consider several fundamental proximity problems, which we define in more detail below. We design KDS's with better performance for some these problems, and we provide the first kinetic results for others. We introduce a simple method that underlies all these results. We briefly describe the approach in Section~\ref{sec:ourApproach}.

Finding the nearest point in $P$ to a query point is called the \textit{nearest neighbor search} problem (or the \textit{post office} problem), and is a well-studied proximity problem.  The \textit{all nearest neighbors} problem, a variant of the nearest neighbor search problem, is to find the nearest neighbor $q\in P$ to each point $p\in P$. The directed graph constructed by connecting each point $p$  to its nearest neighbor $q$ with a directed edge $\overrightarrow{pq}$ is called the \textit{nearest neighbor graph} (NNG). The \textit{closest pair} problem is to find a pair of points in $P$ whose separation distance is minimum; the endpoints of the edge(s) with minimum length in the nearest neighbor graph give the closest pair. For the set $P$, there exists a complete, edge-weighted graph $G(V,E)$ where $V=P$ and the weight of each edge is the distance between its two endpoints in the Euclidean metric.  

A \emph{Euclidean minimum spanning tree} (EMST) of $G$ is a connected subgraph of $G$ such that the sum of the edge weights in the Euclidean metric is minimum possible. The Yao graph~\cite{DBLP:journals/siamcomp/Yao82} and the Semi-Yao graph (or theta graph)~\cite{Clarkson:1987:AAS:28395.28402,Keil:1988:ACE:61764.61787} of a point set $P$ are two well-studied sparse proximity graphs. Both of these graphs are constructed in the following way. At each point $p\in P$, the plane is partitioned into $z$ wedges $W_0(p),...,W_{z-1}(p)$ with equal apex angles $2\pi/z$. Then for each wedge $W_i(p)$, $0\leq i\leq z-1$, the apex $p$ is connected to a particular point $q\in P\cap W_i(p)$. In the Yao graph, the point $q$ is the point in $P\cap W_i(p)$ with the minimum Euclidean distance to $p$; in the Semi-Yao graph, the point $q$ is the point in $P\cap W_i(p)$ with minimum length projection on the bisector of $W_i(p)$. From now on, unless stated otherwise, when we consider the Yao graph or the Semi-Yao graph, we assume $z=6$.

With these definitions in mind, in Section~\ref{sec:ourApproach} we describe our approach. Before we can describe the main contributions and the kinetic results we obtain using our simple method, we need to review both the terminology of the KDS framework, which is described in Section~\ref{sec:KDSframework}, as well as the previous results, which are described in Section~\ref{sec:relatedwork}.
\subsection{Our Approach}\label{sec:ourApproach}
We provide a new, simple, and deterministic method for maintenance of  all the nearest neighbors, the closest pair,  the Euclidean minimum spanning tree (EMST or $L_2$-MST), the Yao graph, and the Semi-Yao graph. In particular, to the best of our knowledge our KDS's for these graphs are the first KDS's. 

The heart of our approach is to define, compute, and kinetically maintain supergraphs for the Yao graph and the Semi-Yao graph. Then we take advantage of the fact that (as we explain later) these graphs are themselves supergraphs of the EMST and the nearest neighbor graph, respectively. 

We define a supergraph for the Yao graph as follows. We partition a unit disk into six \textquotedblleft pieces of pie\textquotedblright~$\sigma_0,\sigma_2,...,\sigma_5$ with equal angles such that all $\sigma_l$, $l=0,...,5$, share a point at the center of the disk (see Figure~\ref{fig:disk_hexagon}(a)). Each piece of pie $\sigma_l$ is a convex shape. For each $\sigma_l$ we construct a triangulation as follows. Using the fact that, for a set $P$ of points, a Delaunay triangulation can be defined based on any convex shape~\cite{Chew:1985:VDB:323233.323264,Drysdale:1990:PAC:320176.320194}, we define a Delaunay triangulation $DT_l$ based on each piece of pie $\sigma_l$. The union of all of these Delaunay triangulations $DT_l$, $l=0,...,5$, which we call the \textit{Pie Delaunay graph}, is a supergraph of the Yao graph. Since the Yao graph, for $z\geq 6$, is guaranteed to contain the EMST, the Pie Delaunay graph contains the EMST.

We define a supergraph for the Semi-Yao graph as follows. 
We partition a hexagon into six equilateral triangles $\Delta_0,\Delta_2,...,\Delta_5$ (see Figure~\ref{fig:disk_hexagon}(b)), and for each equilateral triangle $\Delta_l$ we define a Delaunay triangulation $DT_l$. The union of all of these Delaunay triangulations $DT_l$, $l=0,...,5$, which we call the \textit{Equilateral Delaunay graph}, is a supergraph of the Semi-Yao graph. We prove that the Semi-Yao graph is a supergraph of the nearest neighbor graph, which implies that the Equilateral Delaunay graph is a supergraph of the nearest neighbor graph.

\begin{figure}[t!]
  \begin{center}
    \includegraphics[scale=1]{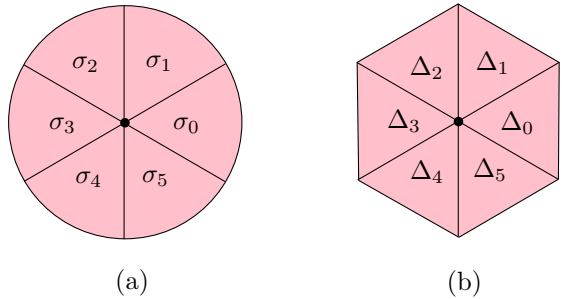}
  \end{center}
  \caption{(a) Partitioning a unit disk into six pieces of pie. (b) Partitioning a hexagon into six equilateral triangles.}
  \label{fig:disk_hexagon}
\end{figure}

In the case that the Delaunay triangulation $DT_l$ is based on a piece of pie, the triangulation can easily be maintained over time. This leads us to a kinetic data structure for the union of the $DT_l$'s, \ie, the Pie Delaunay graph. Then we show how to use this sparse graph over time to give kinetic data structures for maintenance of  the Yao graph and the EMST. Similarly, in the case that each $DT_l$ arises from an equilateral triangle, we obtain a kinetic data structure for the Equilateral Delaunay graph. Using the kinetic Equilateral Delaunay graph we give kinetic data structures for maintenance of the Semi-Yao graph, all the nearest neighbors, and the closest pair.
\subsection{KDS Framework}\label{sec:KDSframework}
Basch, Guibas and Hershberger~\cite{basch_data_1999} first introduced the \textit{kinetic data structure} (KDS) framework to maintain \textit{attributes}, \eg, the closest pair, of a set of $n$ moving points. This approach has been used extensively to model motion. They introduced four standard criteria to evaluate the performance of a KDS: \textit{efficiency}, \textit{responsiveness}, \textit{compactness}, and \textit{locality}.

In the KDS framework, one defines a set of \textit{certificates} that together attest that the desired attribute holds throughout intervals of time between certain \textit{events}, described below. A certificate is a Boolean function of time, and it may have a failure time $t$. The certificate is valid until time $t$. A \textit{priority queue} of the failure times of the certificates is used to track the first time after the current time $t_c$ that a certificate will become invalid. When the failure time of a certificate with highest priority in the queue is equal to the current time $t_c$, the certificate fails, and we say that an \textit{event} occurs. Then we invoke an update mechanism to replace the certificates that become invalid with new valid ones, and apply the necessary changes to the data structures. 

Now we describe the four performance criteria: 
\begin{itemize}
\item[1.] \textit{Responsiveness:} One of the most important KDS performance criteria is the processing time to handle an event. The KDS is \textit{responsive} if the response time of the update mechanism for an event is $O(\log^c n)$; $n$ is the number of points and $c$ is a constant.

\item[2.] \textit{Compactness:} The compactness criterion concerns the total number of certificates stored in the KDS at any given time. If the number of certificates is $O(n\log^c n)$, the KDS is \textit{compact}.

\item[3.] \textit{Locality:} If the number of certificates associated with a particular point is $O(\log^c n)$, the KDS is \textit{local}. Satisfaction of this criterion ensures that, for any point, if it changes its trajectory it participates in a small number of certificates, and therefore, only a small number of changes are needed in the  KDS.

\item[4.] \textit{Efficiency:} To count the number of events over time we make the assumption that the trajectories of the points are polynomial functions of bounded degree $s$. The efficiency of a KDS concerns the number of events in the KDS over time. To analyse the efficiency of a KDS one identifies two types of events. Some events do not necessarily change the attribute of interest (also called  the \textit{desired} attribute) and may only change some internal data structures. Such events are called \textit{internal events}. Those events that change the attribute of interest are called \textit{external events}. If the ratio between the number of internal events and the number of external events is  $O(\log^c n)$, the KDS is \textit{efficient}. The efficiency of a KDS can be viewed as measuring the fraction of events that are due to overhead. 
\end{itemize}
\subsection{Other Related Work}\label{sec:relatedwork}
\paragraph{Kinetic All Nearest Neighbors.}
The nearest neighbor graph is a subgraph of the Delaunay triangulation and the Euclidean minimum spanning tree. Thus by maintaining either one of these supergraphs over time,  all the nearest neighbors can also be maintained. In particular, by using the kinetic Delaunay triangulation~\cite{Albers_voronoidiagrams} or the kinetic Euclidean minimum spanning tree~\cite{DBLP:conf/iwoca/RahmatiZ11}, together with a basic tool in the KDS framework called the kinetic tournament tree~\cite{basch_data_1999}, we can maintain  all the nearest neighbors over time. For both these two approaches, the number of  internal events is nearly cubic in $n=|P|$. Since the  number of  external events for  all the nearest neighbors is nearly quadratic, neither of these two approaches will give an efficient KDS as defined above. 

Agarwal, Kaplan, and Sharir~\cite{Agarwal:2008:KDD:1435375.1435379} presented the first efficient KDS for maintenance of  all the nearest neighbors. For a set of points in the plane, their kinetic algorithm uses a $2$-dimensional \textit{range tree}. To bound the number of events in order to obtain an efficient KDS, they implemented the range tree by randomized search trees (\textit{treaps}). Their randomized kinetic approach uses $O(n\log^2 n)$ space and processes $O(n^2\beta_{2s+2}^2(n)\log^3 n)$ events, where $\beta_{s}(n)$ is an extremely slow-growing function. The expected time to process all events is $O(n^2\beta_{2s+2}^2(n)\log^4 n)$. In terms of the KDS performance criteria, their KDS is \textit{efficient}, \textit{responsive} (in an amortized sense), and \textit{compact}, but it is not \textit{local}.
\paragraph{Kinetic Closest Pair.}
For a set of points moving in $\mathbb{R}^2$, Basch, Guibas, and Hershberger~\cite{Basch:1997:DSM:314161.314435} presented a KDS to maintain the closest pair. Their kinetic algorithm uses $O(n)$ space and processes $O(n^2\beta_{2s+2}(n)\log n)$ events, each in $O(\log^2 n)$ time; their KDS is responsive, efficient, compact, and local. 

Basch, Guibas, and Zhang~\cite{Basch:1997:PPM:262839.262998} used a multidimensional range tree to maintain the closest pair. Their KDS uses $O(n\log n)$ space and processes $O(n^2\beta_{2s+2}(n)\log n)$ events, each in worst-case time $O(\log^2 n)$. Their KDS, which can be used for higher dimensions as well, is responsive, efficient, compact, and local. The same KDS with the same complexities as~\cite{Basch:1997:PPM:262839.262998} was independently presented by Agarwal, Kaplan, and Sharir~\cite{Agarwal:2008:KDD:1435375.1435379}; the KDS by Agarwal~\etal~supports point insertions and deletions.
\paragraph{Kinetic EMST.}
Fu and Lee~\cite{Fu:1991:MST:115128.115142} proposed the first kinetic algorithm for maintenance of an EMST on a set of $n$ moving points. Their algorithm uses $O(sn^4\log n)$ preprocessing time and $O(m)$ space, where $m$ is the maximum possible number of changes in the EMST from time $t=0$ to $t=\infty$. At any given time, the algorithm constructs the EMST in linear time.

Agarwal~\etal~\cite{Agarwal:1998:PKM:795664.796402} proposed a sophisticated algorithm for a restricted kinetic version of the EMST over a graph where the distance between each pair of points in the graph is defined by a linear function of time. The processing time for each combinatorial change in the EMST is $O(n^{2\over 3}\log^{4\over 3} n)$; the bound reduces to $O(n^{1\over 2}\log^{3\over 2} n)$ for planar graphs. Their data structure does not explicitly bound the number of changes, but a bound of $O(n^4)$ is easily seen.

For any $\epsilon>0$, Basch, Guibas, and Zhang~\cite{Basch:1997:PPM:262839.262998} presented a KDS for a $(1+\epsilon)$-EMST whose total weight is within a factor of $(1+\epsilon)$ of the total weight of an exact EMST. For a set of points in the plane, their KDS uses  $O(\epsilon^{-1\over 2}n\log n)$ space and $O(\epsilon^{-1\over 2}n\log n)$ preprocessing time, and processes $O(\epsilon^{-1}n^3)$ events, each in $O(\log^2 n)$ time; their KDS works for higher dimensions. They claim that their structure can be used to maintain the minimum spanning tree in the $L_1$ and $L_\infty$ metrics. 

Rahmati and Zarei~\cite{DBLP:conf/iwoca/RahmatiZ11} improved the previous result by Fu and Lee~\cite{Fu:1991:MST:115128.115142}. In particular, Rahmati and Zarei presented an exact kinetic algorithm for maintenance of the EMST on a set of $n$ moving points in $\mathbb{R}^2$. In $O(n\log n)$ preprocessing time and $O(n)$ space, they build a KDS  that processes $O(n^4)$ events, each in $O(\log^2 n)$ time. Their KDS uses the method of Guibas~\etal~\cite{conf/wg/GuibasM91} to track changes to the Delaunay triangulation, which is a supergraph of the EMST~\cite{O'Rourke:1998:CGC:521378}. Whenever two edges of the Delaunay triangulation swap their length order, their kinetic algorithm makes the required changes to the EMST. In fact, under an assumption we will explain soon, the number of changes in their algorithm is within a linear factor of the number of changes to the Delaunay triangulation~\cite{conf/wg/GuibasM91}. Rubin~\cite{DBLP:journals/dcg/Rubin13} proved that the number of discrete changes to the Delaunay triangulation is $O(n^{2+\epsilon})$, for any $\epsilon >0$, under the assumptions that ($i$) any four points can be co-circular at most twice, and ($ii$) either no ordered triple of points can be collinear more than once, or no triple of points can be collinear more than twice. Under these assumptions, the kinetic algorithm of Rahmati and Zarei processes $O(n^{3+\epsilon})$ events, which is within a linear factor of the number of changes to the Delaunay triangulation.

The kinetic approach by Rahmati and Zarei~\cite{DBLP:conf/iwoca/RahmatiZ11} can maintain the minimum spanning tree of a planar graph whose edge weights are polynomial functions of bounded degree; the processing time of each event is $O(\log^2 n)$.

\paragraph{Kinetic Yao graph and Semi-Yao graph.} To the best of our knowledge there are no previous kinetic data structures for maintenance of the Semi-Yao graph and the Yao graph on a set of moving points.
\subsection{Main Contributions and Results}\label{sec:mainResults}
Based on the approach we described in Section~\ref{sec:ourApproach}, we obtain the results below.

\paragraph{Kinetic All Nearest Neighbors and the Closest Pair.}
We give a simple and deterministic kinetic algorithm for maintenance of  all the nearest neighbors of a set $P$ of $n$ moving points in the plane, where the trajectory of each point is a polynomial function of at most constant degree $s$. Our KDS uses linear space and $O(n\log n)$ preprocessing time to construct the kinetic data structure, and processes $O(n^2\beta^2_{2s+2}(n)\log n)$ events with total processing time $O(n^2\beta^2_{2s+2}(n)\log^2 n)$. 

We also show how to maintain the closest pair over time. Our KDS for maintenance of the closest pair has the same complexities as the KDS for  all the nearest neighbors; in particular, it uses $O(n)$ space and processes $O(n^2\beta^2_{2s+2}(n)\log n)$ events for a total processing time of $O(n^2\beta^2_{2s+2}(n)\log^2 n)$.

Our KDS for the all nearest neighbors and the closest pair problems is efficient, responsive in an amortized sense, and compact.
The compactness of the KDS implies that our KDS is local in an amortized sense. In particular, on average each point in our KDS participates in $O(1)$  certificates.

Our \textit{deterministic} algorithm for maintenance of all the nearest neighbors in $\mathbb{R}^2$ is simpler and more efficient than the \textit{randomized} kinetic algorithm by Agarwal, Kaplan, and Sharir~\cite{Agarwal:2008:KDD:1435375.1435379}: both of these kinetic algorithms need a priority queue containing all certificates of the KDS (our priority queue uses linear space, but their priority queue uses $O(n\log^2 n)$ space). Our KDS uses a graph data structure for the Equilateral Delaunay graph and a constant number of tournament trees for each point, but their KDS uses a \textit{$2$d range tree} implemented by randomized search trees (treaps), a constant number of sorted lists, and in fact it maintains  $O(\log^2 n)$ tournament trees for each point. In particular,
\begin{itemize}
\item we perform one-dimensional range searching, as opposed to the two-dimensional range searching of their work;
\item the sparse graph representation allows us to obtain a linear space KDS, which improves the space complexity $O(n\log^2 n)$ of their KDS. Their KDS uses a $2$d range tree implemented by randomized search trees that in effect maintain a supergraph of the nearest neighbor graph with $O(n\log^2 n)$ candidate edges;
\item in our kinetic algorithm, the number of changes to the Equilateral Delaunay graph when the points are moving is $O(n^2\beta_{2s+2}(n))$; this leads us to have total processing time $O(n^2\beta^2_{2s+2}(n)\log^2 n)$, which is an improvement of the total expected processing time $O(n^2\beta^2_{2s+2}(n)\log^4 n)$ of their randomized algorithm;
\item on average each point in our KDS participates in a constant number of certificates, but each point in their KDS participates in $O(\log^2 n)$ certificates.
\end{itemize}

The certificates of our KDS for maintenance of the closest pair are simpler than the certificates of the previous kinetic algorithms by Basch, Guibas, and Hershberger (SODA'97)~\cite{Basch:1997:DSM:314161.314435}, Basch, Guibas, and Zhang (SoCG'97)~\cite{Basch:1997:PPM:262839.262998}, and Agarwal, Kaplan, and Sharir  (TALG 2008)~\cite{Agarwal:2008:KDD:1435375.1435379}.

\paragraph{Kinetic Yao Graph and Semi-Yao Graph.} We give the first kinetic data structures for maintenance of two well-studied sparse graphs, the Semi-Yao graph and the Yao graph. Our KDS processes $O(n^2\beta_{2s+2}(n))$ (resp. $O(n^3\beta_{2s+2}^2(n)\log n)$) events to maintain the Semi-Yao graph (resp. the Yao graph); each event can be processed in time $O(\log n)$ in an amortized sense.

\paragraph{Kinetic EMST.} Our KDS for maintenance of the EMST uses $O(n)$ space, takes $O(n\log n)$ preprocessing time, and processes $O(n^3\beta^2_{2s+2}(n)\log n)$ events.  The total cost to process all these events is $O(n^3\beta^2_{2s+2}(n)\log^2 n)$. Our KDS is responsive in an amortized sense, compact, and local on average.

Our EMST KDS improves on the previous EMST KDS by Rahmati and Zarei~\cite{DBLP:conf/iwoca/RahmatiZ11}. Our KDS processes $O(n^3\beta^2_{2s+2}(n)\log n)$ events, whereas the KDS by Rahmati and Zarei processes $O(n^4)$ events.

Table~\ref{table:RelatedWork} summarizes our results and compares them with the previous results.

\begin{table}[t]
\small
\centering
\begin{tabular}{ | c | p{2cm} || c | c | c | p{1.5cm} |}
\cline{2-6} 
\multicolumn{1}{ c| }{} & 
problem & space  &  total number of events &  proc. time per event&  locality\\ \cline{1-6} 
  \multirow{2}{*}{Basch~\etal~\cite{basch_data_1999}} & \multirow{2}{*}{closest pair} &  \multirow{2}{*}{$O(n)$} & \multirow{2}{*}{$O(n^2\beta_{2s+2}(n)\log n)$} & \multirow{2}{*}{$O(\log^2 n)$ ~~~~[in wrc]} & $O(\log n)$ in wrc\\  \hline\hline
\multirow{4}{*}{Basch~\etal~\cite{Basch:1997:PPM:262839.262998}} &  \multirow{2}{*}{closest pair}& \multirow{2}{*}{$O(n\log n)$} & \multirow{2}{*}{$O(n^2\beta_{2s+2}(n)\log n)$} & \multirow{2}{*}{$O(\log^2 n)$ ~~~~[in wrc]} & $O(\log n)$ in wrc\\  \cline{2-6} 
& \multirow{2}{*}{$(1+\epsilon)$-EMST}& \multirow{2}{*}{$O(\epsilon^{-1\over 2}n\log n)$} & \multirow{2}{*}{$O(\epsilon^{-1}n^3)$} & \multirow{2}{*}{$O(\log^2 n)$ ~~~~[in wrc]} & $O(\log n)$ in wrc\\  \hline\hline
\multirow{4}{*}{Agarwal \etal~\cite{Agarwal:2008:KDD:1435375.1435379}} &  \multirow{2}{*}{closest pair}& \multirow{2}{*}{$O(n\log n)$} & \multirow{2}{*}{$O(n^2\beta_{2s+2}(n)\log n)$} & \multirow{2}{*}{$O(\log^2 n)$ ~~~~[in wrc]} & $O(\log n)$ in wrc\\  \cline{2-6} 
& all nearest neighbors& \multirow{2}{*}{$O(n\log^2 n)$} & \multirow{2}{*}{$O(n^2\beta_{2s+2}^2(n)\log^3 n)$} & \multirow{2}{*}{$O(\log n)$ ~~~~[in amr]} & $O(\log^2 n)$ on avg\\  \hline\hline
  \multirow{2}{*}{Rahmati~\etal~\cite{DBLP:conf/iwoca/RahmatiZ11}} & \multirow{2}{*}{EMST} &  \multirow{2}{*}{$O(n)$} & \multirow{2}{*}{$O(n^4)$} & \multirow{2}{*}{$O(\log^2 n)$ ~~~[in wrc]} & $O(1)$ ~on avg\\  \hline\hline
\multirow{10}{*}{{{\color{Mahogany}This Paper}}} & \multirow{2}{*}{closest pair} & \multirow{2}{*}{$O(n)$} & \multirow{2}{*}{$O(n^2\beta_{2s+2}(n)\log n)$} & \multirow{2}{*}{$O(\log n)$ ~~~~[in amr]} & $O(1)$ on ~avg\\ \cline{2-6} 
& all nearest neighbors& \multirow{2}{*}{$O(n)$} & \multirow{2}{*}{$O(n^2\beta_{2s+2}(n)\log n)$} & \multirow{2}{*}{$O(\log n)$ ~~~~[in amr]} & $O(1)$ on ~avg \\ \cline{2-6} 
& \multirow{2}{*}{EMST} & \multirow{2}{*}{$O(n)$} & \multirow{2}{*}{$O(n^3\beta^2_{2s+2}(n)\log n)$} & \multirow{2}{*}{ $O(\log n)$ ~~~~[in amr]} & $O(1)$ on ~avg \\ \cline{2-6} 
& \multirow{2}{*}{Yao graph} & \multirow{2}{*}{$O(n)$} & \multirow{2}{*}{$O(n^3\beta^2_{2s+2}(n)\log n)$} & \multirow{2}{*}{$O(\log n)$ ~~~~[in amr]} & $O(1)$ on ~avg\\   \cline{2-6} 
& Semi-Yao graph & \multirow{2}{*}{$O(n)$} & \multirow{2}{*}{$O(n^2\beta_{2s+2}(n))$} & \multirow{2}{*}{$O(\log n)$ ~~~~[in amr]} & $O(1)$ on ~avg\\\hline
\end{tabular}
\vspace{+10pt}
\caption{The comparison between our KDS's and the previous KDS's, for a set of $n$ points in the plane. The abbreviations amr, wrc, and avg stand for amortized, worst-case, and average, respectively.}
\label{table:RelatedWork}
\vspace{-20pt}
\end{table}
\subsection{Organization}
As necessary background for our work, Section~\ref{sec:preliminary} reviews a basic tool, the \textit{kinetic tournament tree}, which is used in the kinetic data structure framework.

Section~\ref{sec:ANN_CP} is organized as follows: Subsection~\ref{sec:ANN_CP_construction} gives the new method for computing all the nearest neighbors and the closest pair. In particular, it introduces our two new sparse graphs, the \textit{Semi-Yao graph} and the \textit{Equilateral Delaunay graph}  (in fact we will show these graphs are the same). In Subsection~\ref{sec:kinetic_EDG}, we make a kinetic version of the Equilateral Delaunay graph, and then in Subsections~\ref{sec:kinetic_ANN} and~\ref{sec:kinetic_CP}, we show how to use it to maintain  all the nearest neighbors and the closest pair.

The organization of Section~\ref{sec:YG_EMST} is similar to that of Section~\ref{sec:ANN_CP}. Using a new sparse graph, which we call the Pie Delaunay graph, we provide our new method for constructing the Yao graph and the EMST in Subsection~\ref{sec:YG_EMST_construction}. Subsection~\ref{sec:kinetic_PDG} gives a KDS for maintenance of the Pie Delaunay graph, and Subsections~\ref{sec:kinetic:YG} and \ref{sec:kinetic_EMST} use this KDS to maintain the Yao graph and the EMST.

Section~\ref{sec:conclusion} discusses the extensions of the presented kinetic data structures to higher dimensions and gives some open problems for continuing this research direction.
\section{Preliminaries}\label{sec:preliminary}
Let ${\cal O}=\{o_1,o_2,...,o_n\}$ be a set of $n$ moving objects in the plane, where the $y$-coordinate $y_i(t)$ of each object $o_i$ is a continuous function of time. Assuming  $y_i(t)$ is a polynomial function of at most constant degree $s$, it follows from Theorem~\ref{the:totallyDFcomplexity} below that the number of all changes for the lowest object with respect to the $y$-axis, among the set of objects ${\cal O}$, is $\lambda_s(n)$. 
\begin{theorem}\label{the:totallyDFcomplexity}{\tt \cite{Agarwal:1995:DSG:868483}}
The length of the lower envelope of $n$ totally-defined, continuous, univariate  functions, such that each pair of them intersects at most $s$ times, is at most $\lambda_s(n)$.
\end{theorem}
Note that Theorem~\ref{the:totallyDFcomplexity} holds for totally-defined functions; there exists a similar result for partially-defined functions:
\begin{theorem}\label{the:partiallyDFcomplexity}{\tt \cite{Agarwal:1995:DSG:868483}}
The length of the lower envelope of $n$ partially-defined, continuous, univariate  functions, such that each pair of them intersects at most $s$ times, is at most $\lambda_{s+2}(n)$.
\end{theorem}
Here, $\lambda_s(n)=n\beta_s(n)$ is the maximum length of Davenport-Schinzel sequences of order $s$ on $n$ symbols, and $\beta_s(n)$ is an extremely slow-growing function. In particular, 
\[
\lambda_s(n) = 
\begin{cases}
        {n},  & \text{for $s=1$ };\\
        {2n-1},  & \text{for $s=2$ };\\
        {2n\alpha(n)+O(n)},  & \text{for $s=3$ };\\
        {\Theta(n2^{\alpha(n)})},  & \text{for $s=4$ };\\
        {\Theta(n\alpha(n)2^{\alpha(n)})},  & \text{for $s=5$ };\\
        {n2^{(1+o(1))\alpha^t(n)/t!}}, & \text{for $s\geq 6$};
\end{cases}
\]
here $t={\lfloor {(s-2)/2}\rfloor}$ and  $\alpha(n)$ denotes the inverse Ackermann function~\cite{Pettie:2013:SBD:2493132.2462390}. 

For maintenance of the lowest object with respect to the $y$-axis among the set of moving objects $\cal O$ over time, we use a basic (kinetic) data structure called a \textit{kinetic tournament tree}~\cite{basch_data_1999,Agarwal:2008:KDD:1435375.1435379}. A kinetic tournament tree is a balanced binary tree $T$ such that the objects are stored at the leaves of the tree $T$ in an arbitrary order, and each internal node $v$ of the tree maintains the lowest object between its two children. In more detail, denote by  $T_v$ the subtree rooted at internal node $v$ and denote by $P_v$ the set of objects stored at the leaves of $T_v$. The object stored at $v$ in the tournament tree is the lowest object among all the objects in $P_v$; this object is called the \textit{winner} of the subtree $T_v$. For each internal node $v$ of the tournament tree we define a \textit{certificate} to assert whether the left-winner (winner of the left subtree) or the right-winner (winner of the right subtree) is the winner for $v$. The failure time of the certificate corresponding to the internal node $v$ is the time when the winner at $v$ changes. All of the certificates together are stored in a \textit{priority queue}, with the failure times as the keys, to track the next time after the current time that a certificate will become invalid.

When the certificate corresponding to an internal node $v$ fails, it may change some winners on the path from the parent of $v$ to the root. In some cases the winner of a node $v'$ on the path does not change, but the failure time corresponding to the certificate of the node $v'$ may change. Therefore, we must update the failure times of the certificates of the nodes on the path from the parent of $v$ to the root, and then we must replace the invalid certificates with new valid ones in the priority queue; this takes $O(\log^2 n)$ time, which implies that the KDS is \textit{responsive}. The number of internal events for all the internal nodes is $\sum_v\lambda_s(|P_v|)= O(\lambda_s(n)\log n)$. Since the number of external events, that is the number of changes to the root of the tournament tree, is $\lambda_s(n)$, the KDS is \textit{efficient}. The tournament tree uses linear space, which implies the KDS is \textit{compact}. Each object participates in $O(\log n)$ certificates, which means the KDS is \textit{local}.

It is convenient for our purpose to make the tournament tree dynamic, to support point insertions and deletions; the dynamic version of the kinetic tournament tree is called a \textit{dynamic and kinetic tournament tree}. This dynamic and kinetic tournament tree can be implemented using a \textit{weight-balanced (BB($\alpha$)) tree}~\cite{NJRE:10.1137/0202005,Mehlhorn:1984}; see the construction of a dynamic and kinetic tournament tree in~\cite{Agarwal:2008:KDD:1435375.1435379}. Consider a sequence of $m$ insertions and deletions into a dynamic and kinetic tournament tree where the maximum size  tree at any time is $n$ (assuming $m\geq n$). The following theorem gives the construction time and the processing time of  a dynamic and kinetic tournament tree.
\begin{theorem}\label{the:DynamicKineticTT}{\tt \cite{Agarwal:2008:KDD:1435375.1435379}}
A dynamic and kinetic tournament tree on $n$ elements can be constructed in $O(n)$ time. The tournament tree generates at most  $O(m\beta_{s+2}(n)\log n)$
events, for a total cost of $O(m\beta_{s+2}(n)\log^2 n)$. Processing an event takes $O(\log^2 n)$ time.
\end{theorem}
\section{All Nearest Neighbors and Closest Pair}\label{sec:ANN_CP}
In this section we provide a sparse graph representation and show a new construction of the nearest neighbor graph. First, we introduce two new supergraphs of the nearest neighbor graph, namely the \textit{Semi-Yao graph} and the\textit{ Equilateral Delaunay graph} (EDG), and then we show that these graphs are in fact the same. Next, we show how to maintain the Equilateral Delaunay graph for moving points, and then we give  simple KDS's for maintenance of all the nearest neighbors and the closest pair.
\subsection{New Method for Computing All Nearest Neighbors and Closest Pair}\label{sec:ANN_CP_construction}
Partition the plane into six \textit{wedges} (cones) $W_0,...,W_{5}$, each of angle $\pi/3$ with common apex at the origin $o$. For $0\leq l \leq 5$, let $W_l$ span the angular range $[(2l-1)\pi/6, (2l+1)\pi/6)$. Denote by $b_l$ the unit vector in the direction of the bisector ray of $W_l$. Let $W_l(p_i)$ denote the translate of wedge $W_l$ that moves the apex to point $p_i$, and let ${\cal V}_l(p_i)$ denote the intersection of $P$ with wedge $W_l(p_i)$: ${\cal V}_l(p_i)=P\cap W_l(p_i)$. Denote by $b_l(p_i)$ the unit vector emanating from $p_i$ in the direction of the bisector ray of $W_l(p_i)$; see Figure~\ref{fig:projection}(a). Observe that, in Figure~\ref{fig:projection}(a), since $p_i$ is the closest point to $p_j$, there are no other points of $P$ in the interior of the disc. Let $d(p_i,p_j)$ denote the distance between points $p_i$ and $p_j$.

\begin{figure}[t!]
\centering
\includegraphics[scale=1]{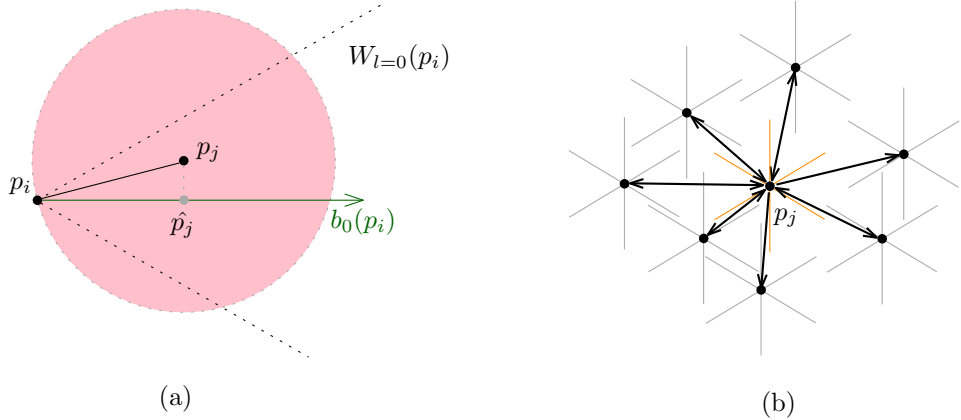}
\caption{ (a) Projection of the point $p_j$ to the bisector $b_0(p_i)$ of the wedge $W_0(p_i)$. (b) In-edges and out-edges of $p_j$.}
\label{fig:projection}
\end{figure}

The following straightforward lemma is key for obtaining our
kinetic data structure for the \textit{all nearest neighbors} and the
\textit{closest pair} problems. Consider $p_j\in P$, and let $p_i$
denote the point of $P$ closest to $p_j$ and distinct from $p_j$.
Let $W_l(p_i)$ denote the wedge of $p_i$ that contains $p_j$, and
denote by $\hat{p}_j$ the projection of $p_j$ to the bisector
$b_l(p_i)$ (see Figure~\ref{fig:projection}(a)).

\begin{lemma}\label{the:CPidea}{\tt \cite{Agarwal:2008:KDD:1435375.1435379,basch_data_1999}}
Point $p_j$ has the minimum length projection to $b_l(p_i)$, where the minimum is taken over ${\cal V}_l(p_i)$. That is,
\begin{equation}\label{eq:MinCoordinate}
d(\hat{p}_j,p_i)=\min\{d(\hat{p}_k,p_i) | p_k\in {\cal V}_l(p_i)\}.
\end{equation}
\end{lemma}

Thus, Lemma~\ref{the:CPidea} gives a necessary condition for $p_i$ to be the nearest neighbor to $p_j$. We now use this lemma to define a super-graph of the nearest neighbor graph of $P$. 
To find the nearest neighbor for each point $p_j\in P$,  we seek a
set of candidate points ${\cal
C}(p_j)=\{p_i|~p_i~and~p_j~satisfy~Equation~(1)\}$. From now on, when we say $p_j$ has the minimum $b_l$-coordinate inside
the wedge $W_l(p_i)$, we mean that $p_j$ and $p_i$ satisfy
Equation~(\ref{eq:MinCoordinate}).

By connecting each point $p_i\in P$ to a  point $p_j\in {\cal
V}_l(p_i)$ with a directed edge $\overrightarrow{p_jp_i}$ from $p_j$ to $p_i$ whenever $p_j$
is the point with the minimum $b_l$-coordinate, among all the points
in ${\cal V}_l(p_i)$, we obtain what we call the \textit{Semi-Yao graph}
(SYG) of $P$~\footnote{This graph is called the $\theta_6$-graph in~\cite{DBLP:journals/dcg/KeilG92}, but we prefer to call it the Semi-Yao graph instead of the $\theta_6$-graph, because of its close relationship to the Yao graph~\cite{DBLP:journals/siamcomp/Yao82}}. The edge $\overrightarrow{p_jp_i}$ is called an
\textit{in-edge} for $p_i$ and it is called an \textit{out-edge} for
$p_j$. Each point in the Semi-Yao graph has at most six in-edges and has a
set of out-edges; Figure~\ref{fig:projection}(b) depicts the in-edges
and the out-edges of the point $p_j$. Denote by $S_{out}(p_j)$ the
end points of the  out-edges of $p_j$. From the above discussion, it is easy to see the following observation and lemma.

\begin{observation}
${\cal C}(p_j)=S_{out}(p_j)$.
\end{observation}

\begin{lemma}\label{the:SYcontainsNNG}
The Semi-Yao graph is a super-graph of the nearest neighbor graph.
\end{lemma}

From now on, when we say a convex set is \textit{empty}, we mean it has no point of $P$ in its interior. 

From Lemma~\ref{the:CPidea}, we obtain the following straightforward observation, which makes a connection to the Delaunay triangulations of the point set $P$.
\begin{observation}\label{the:emptyTri}
If $p_j$ has the minimum $b_l$-coordinate inside the wedge $W_l(p_i)$, then $p_i$ and $p_j$ touch the boundary of an empty equilateral triangle; $p_i$ touches a vertex and $p_j$ touches an edge of the triangle.
\end{observation}

A \textit{unit regular hexagon} is a regular hexagon whose edges
have unit length; let $\varhexagon$ be the unit regular hexagon with
center at the origin $o$ and  vertices at $(\sqrt{3}/2,1/2)$,
$(0,1)$, $(-\sqrt{3}/2,1/2)$, $(-\sqrt{3}/2,-1/2)$, $(0,-1)$, and
$(\sqrt{3}/2,-1/2)$ (see Figure~\ref{fig:hexagon}(a)). Partition
$\varhexagon$ into six equilateral triangles $\vartriangle_l$,
$l=0,1,..,5$, and call any translated and scaled copy of
$\vartriangle_l$ an \textit{$l$-tri} (see
Figure~\ref{fig:hexagon}(b)).

\begin{figure}[t]
  \begin{center}
    \includegraphics[scale=1]{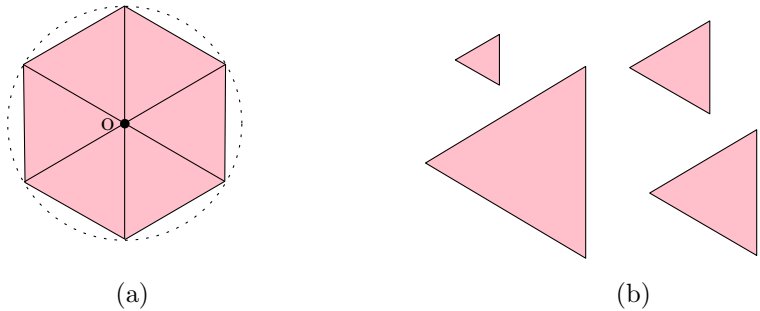}
  \end{center}
  \caption{(a) Partitioning the unit regular hexagon into six equilateral triangles. (b) Some $0$-tri's.}
  \label{fig:hexagon}
\end{figure}

A Delaunay graph can be defined based on any convex shape, \eg, a
square, a diamond, any triangle, or a piece of
pie~\cite{Abam:2010:SEK:1630166.1630284,DBLP:conf/swat/RahmatiZ12,Drysdale:1990:PAC:320176.320194}. The Delaunay triangulation based on a convex shape is the maximal set of edges such that no two edges intersect except at
common endpoints, and such that the endpoints of each edge lie on the boundary of an empty scaled translate of the convex shape. If the points are in \textit{general position}\footnote{The set of points $P$ is in general position with respect to a convex shape if it contains no four points on
the boundary of any scaled translate of the convex shape.} the bounded faces of the Delaunay
graph are triangles, and the Delaunay graph is called a Delaunay
triangulation. Here we call the Delaunay triangulation constructed based on an equilateral triangle an \textit{Equilateral Delaunay triangulation} (EDT).

There is a nice connection between the Semi-Yao graph and Equilateral Delaunay triangulations. In general, the Semi-Yao graph is the union of two Equilateral Delaunay triangulations~\cite{Bonichon:2010:CTD:1939238.1939265}. Next we describe this connection in a different, and in our view simpler, way than~\cite{Bonichon:2010:CTD:1939238.1939265}. 

Call an $l$-tri whose interior does not contain any point of $P$ an \textit{empty $l$-tri}. Denote by $EDT_l$ the Equilateral Delaunay triangulation based on the $l$-tri. The
edge $p_ip_j$ is an edge of $EDT_l$ if and only if there is an empty
$l$-tri such that $p_i$ and $p_j$ are on the boundary of the $l$-tri;
Figure~\ref{fig:vor_del} depicts $EDT_0$ for a set of four points. Let ${\cal E}(G)$ be the set of edges of graph $G$; the set of vertices of $G$ is
$P$.  Since $\vartriangle_0$, $\vartriangle_2$, and $\vartriangle_4$ are
translates of one another, and similarly for $\vartriangle_1$,
$\vartriangle_3$, and $\vartriangle_5$, we have that ${\cal
E}(EDT_0)={\cal E}(EDT_2)={\cal E}(EDT_4)$ and ${\cal
E}(EDT_1)={\cal E}(EDT_3)={\cal E}(EDT_5)$. Thus, there are two
different types of $l$-tri's. We define the \textit{Equilateral Delaunay graph} (EDG) to be the union
of $EDT_0$ and $EDT_1$, \ie, $p_ip_j\in {\cal E}(EDG)$ if
and only if $p_ip_j\in {\cal E}(EDT_0)$ or $p_ip_j\in {\cal
E}(EDT_1)$.

\begin{figure}[t]
\centering
    \includegraphics[scale=0.4]{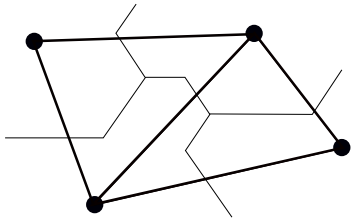}
  \caption{The Delaunay triangulation and the Voronoi diagram based on the $0$-tri, as produced by a program in~\cite{GeometryLabRolf}.}
  \label{fig:vor_del}
\end{figure}

The cell boundaries of a Voronoi diagram of a set $P$ of $n$ sites, based on a convex shape, consist of points where the convex-shaped waves emanating from the sites
collide; to determine the Voronoi diagram of the set of four sites in
Figure~\ref{fig:vor_del}, based on the $0$-tri, we use a program
in~\cite{GeometryLabRolf}. Using divide and conquer algorithms by
Chew and Drysdale~\cite{Chew:1985:VDB:323233.323264,Drysdale:1990:PAC:320176.320194},

\begin{theorem}\label{the:VD_DT_Con}{\tt \cite{Chew:1985:VDB:323233.323264,Drysdale:1990:PAC:320176.320194}}
The Voronoi diagram and Delaunay triangulation of a set of $n$
sites based on a convex shape can be constructed in
$O(n\log n)$ time.
\end{theorem}

Since each $\vartriangle_l$ is a convex shape,
using the approaches of Chew and Drysdale, we can construct the
corresponding Voronoi diagram/Delaunay triangulation  in $O(n\log
n)$ time. Then the following results.

\begin{corollary}\label{the:EDT_Construction}
The  Equilateral Delaunay graph (EDG) can be constructed in $O(n\log n)$ time.
\end{corollary}

Let $p_ip_j\in {\cal E}(EDT_l)$. By definition there exists an empty $l$-tri such that $p_i$ and $p_j$ are on its boundary. By scaling down the $l$-tri, one of the $l$-tri vertices will be placed at $p_i$ or $p_j$; see Figures~\ref{fig:SY_EDT2}(b) and \ref{fig:SY_EDT2}(c).

\begin{observation}\label{the:SqueezedTRI}
If there is an empty $l$-tri such that $p_i$ and $p_j$ are on its boundary, then there is an empty $l$-tri with the same property such that either $p_i$ or $p_j$ is a vertex of the $l$-tri.
\end{observation}

The next lemma proves that the undirected Semi-Yao graph and the Equilateral Delaunay graph are equal to each other.

\begin{lemma}\label{the:SY_EDT}
Edge $p_ip_j\in {\cal E}(SYG)$ if and only if $p_ip_j\in {\cal E}(EDG)$.
\end{lemma}
\begin{proof}
Let $p_ip_j$ be an edge of the undirected Semi-Yao graph such that $p_j$ has the minimum $b_l$-coordinate inside some wedge $W_l(p_i)$ (see
Figure~\ref{fig:SY_EDT2}(a)). The bounded area created by the wedge $W_l(p_i)$ and the line through $p_j$ perpendicular to
$b_l(p_i)$ is an $l$-tri. Therefore, for the edge $p_ip_j$, there exists an empty $l$-tri such that $p_i$ and $p_j$ are on its boundary. This implies that $p_ip_j$ is an edge of  $EDT_l$.

Let $p_ip_j\in {\cal E}(EDT_l)$. By the definition of $EDT_l$, there exists an empty $l$-tri such that $p_i$ and $p_j$ are on its boundary (see Figure~\ref{fig:SY_EDT2}(b)).  By Observation~\ref{the:SqueezedTRI}, that is a rescaled $l$-tri such that $p_i$ and $p_j$ are on its boundary and such that one of the $l$-tri vertices is $p_i$ or $p_j$ (see Figure~\ref{fig:SY_EDT2}(c)); without loss of generality assume it is $p_i$. Point $p_j$ is inside the wedge $W_k(p_i)$, where $k\in\{l,(l+2)\bmod{6}, (l+4)\bmod{6}\}$. Point $p_j$ has the minimum $b_k$-coordinate inside the wedge $W_k(p_i)$; otherwise, there would be a point of $P$ inside the rescaled $l$-tri, which means that $p_ip_j\notin {\cal E}(EDT_l)$, a contradiction. Therefore, $p_ip_j\in {\cal E}(SYG)$.
\end{proof}

\begin{figure}[t!]
\centering \includegraphics[scale=1]{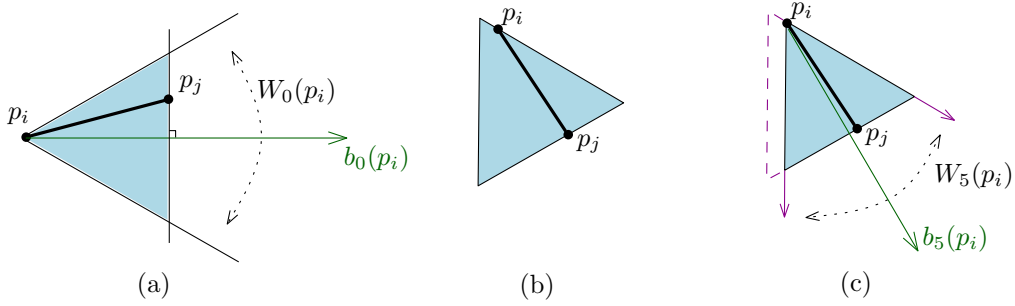} \caption{(a) The point $p_j$ has the minimum $b_0$-coordinate inside the wedge
$W_0(p_i)$. (b) The $1$-tri corresponding to the edge $p_ip_j$ in $EDT_1$ does not contain any other points of $P$. (c) The point $p_j$ is inside the wedge $W_5(p_i)$ and has the minimum $b_5$-coordinate.}
\label{fig:SY_EDT2}
\end{figure}

Now we can give the following result.

\begin{theorem}\label{the:ANN_Construction}
The all nearest neighbors and the closest pair problems in $\mathbb{R}^2$ can be solved in $O(n\log n)$ time.
\end{theorem}
\begin{proof}
From Corollary~\ref{the:EDT_Construction} and Lemma~\ref{the:SY_EDT}, 
the Semi-Yao graph can be constructed in $O(n\log n)$ time. Since the number of edges in the Semi-Yao graph is at most $6n$, by traversing the Semi-Yao graph edges incident to each point, we can find all the nearest neighbors and the closest pair  in linear time. 
\end{proof}
\subsection{Kinetic Equilateral Delaunay Graph}\label{sec:kinetic_EDG}
Since ${\cal E}(EDT_0)={\cal E}(EDT_2)={\cal E}(EDT_4)$ and ${\cal
E}(EDT_1)={\cal E}(EDT_3)={\cal E}(EDT_5)$, to maintain the EDG,
which is the union of $EDT_0$ and $EDT_1$, we need only to have
kinetic data structures for $EDT_0$ and $EDT_1$. We describe how to
maintain $EDT_0$; $EDT_1$ is handled similarly.

The Delaunay triangulation $EDT_0$ is locally stable as long as the
points are in general position. Note that we assume the set of points
$P$ is in general position with respect to a $0$-tri; this means that no four or more points are on the boundary of any scaled, translated $0$-tri. When the points are moving, at a moment $t$ this assumption may fail. In fact for moving points, we make a further assumption: no four points are on
the boundary of the $0$-tri throughout any positive interval of time. This
ensures that the points are in general position over time except at
some discrete moments. The number of these discrete moments over
time is in the order of the number of changes to $EDT_0$, because the failure of the general position assumption is a necessary condition for changing the topological structure of $EDT_0$~\cite{Albers_voronoidiagrams}. When a point moves, $EDT_0$ can change only in the graph neighborhood of the point, and so the correctness of $EDT_0$ over time is asserted by a set of certificates. Our approach for maintenance of $EDT_0$ is a known approach also used in~\cite{Abam:2010:SEK:1630166.1630284,DBLP:conf/swat/RahmatiZ12,Agarwal:2010:KSD:1810959.1810984,Albers_voronoidiagrams}
for maintenance of Delaunay triangulations based on convex shapes.

Figure~\ref{fig:EDT_O}(a)  depicts the $EDT_0$ of a set $P$ of points. Each edge on the boundary of the infinite face of $EDT_0$, like
$p_ip_j$, is called a \textit{hull} edge; the other edges, like
$p_{i'}p_{j'}$, are called \textit{interior} edges. Corresponding to
these two types of edges, we define two types of certificates,
\textit{NotInWedge} and \textit{NotInTri}, respectively. Below, we first we consider the interior edges and then the exterior edges.


\begin{figure}[t!]
\centering \includegraphics[scale=1]{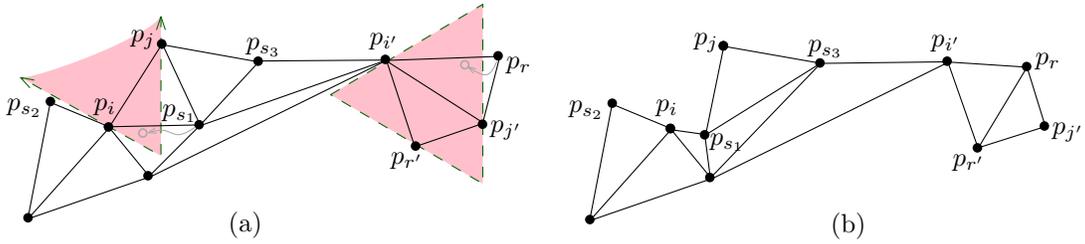} \caption{(a) The
NotInTri certificate corresponding to the edge $p_{i'}p_{j'}$
certifies that $p_r$ is outside the $0$-tri of $p_{i'}$, $p_{j'}$, and
$p_{r'}$. The NotInWedge certificates of the edge $p_ip_j$ certify
that $p_{s_1}$, $p_{s_2}$, and $p_{s_3}$ are outside the
corresponding $k$-wedge. (b) The changes to $EDT_0$  after
$p_r$ moves inside the $0$-tri passing through $p_{i'}$, $p_{j'}$, and
$p_{r'}$ and after $p_{s_1}$ moves inside the $k$-wedge of
$p_ip_j$.} \label{fig:EDT_O}
\end{figure}

\paragraph{Interior Edges.} Each interior edge $p_{i'}p_{j'}\in EDT_0$
is incident to two triangles $p_{i'}p_{j'}p_{r'}$ and
$p_{i'}p_{j'}p_r$ (see Figure~\ref{fig:EDT_O}(a)). For the triangle
$p_{i'}p_{j'}p_{r'}$ (resp. $p_{i'}p_{j'}p_r$), there exists an empty
$0$-tri, denoted by $\Delta^0_{r'}$ (resp. $\Delta^0_r$), such
that $p_{i'}$, $p_{j'}$ and $p_{r'}$ (resp. $p_r$) are on the
boundary of $\Delta^0_{r'}$ (resp. $\Delta^0_r$). For $p_{i'}p_{j'}$, we define a
\textit{NotInTri} certificate certifying that $p_r$ (resp. $p_{r'}$) is outside
$\Delta^0_{r'}$ (resp. $\Delta^0_r$). For sufficiently short time intervals, $p_r$ and $p_{r'}$ are the only points that can change the validity of edge
$p_{i'}p_{j'}$ (see~\cite{Abam:2010:SEK:1630166.1630284,DBLP:conf/swat/RahmatiZ12,Agarwal:2010:KSD:1810959.1810984,Albers_voronoidiagrams}).
Let $t$ be the time when the four points $p_{i'}$, $p_{j'}$,
$p_{r'}$, and $p_r$ are on the boundary of a $0$-tri; at time $t^-$,
$p_r$ (resp. $p_{r'}$) is outside $\Delta^0_{r'}$ (resp.
$\Delta^0_r$). When $p_r$ (resp. $p_{r'}$) moves inside
$\Delta^0_{r'}$ (resp. $\Delta^0_r$), at time $t^+$, this
certificate fails and there is no empty $0$-tri such that $p_{i'}$ and
$p_{j'}$ are on its boundary. Thus at time $t$, we have to delete
the edge $p_{i'}p_{j'}$ and add the new edge $p_{r'}p_r$, because at
time $t^+$ there exists an empty $0$-tri for $p_rp_{r'}$ (see Figure~\ref{fig:EDT_O}(b)). Also, we must define new certificates corresponding to the newly created triangles.

\begin{figure}[t]
  \centering
    \includegraphics[scale=1.2]{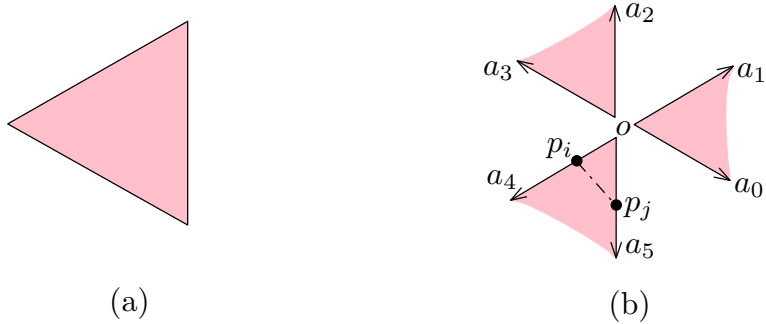}
  \caption{(a) A $0$-tri. (b) The $k$-wedges associated with the $0$-tri; edge $p_ip_j$ divides
  the $4$-wedge $\protect\overleftrightarrow{a_4oa_5}$ into the
   bounded area $\overline{op_ip_j}$ and the unbounded area $\protect\overleftrightarrow{a_4p_ip_ja_5}$.}
  \label{fig:k-wedge}
\end{figure}

\paragraph{Hull Edges.} By removing one of the $0$-tri edges and
extending the other two edges to infinity, three types of wedges are
created; call these wedges \textit{$k$-wedges}, for $k=\{0,2,4\}$,
and denote them by $\overleftrightarrow{a_koa_{k+1}}$ (see
Figure~\ref{fig:k-wedge}); the two sides $\overrightarrow{oa_k}$ and
$\overrightarrow{oa_{k+1}}$ of the boundary of the $k$-wedge are
parallel to the two corresponding sides of the wedge $W_k$. For a hull edge
$p_ip_j$, there exists an empty $k$-wedge such that $p_i$ and $p_j$ are on
the boundary. Each hull edge is incident to at most one triangle
$p_ip_jp_{s_1}$, and adjacent to at most four other hull edges
$p_ip_{s_2}, p_ip_{s_3}, p_jp_{s_4}$ and $p_jp_{s_5}$ on the
boundary cycle of the infinite face; the point $p_{s_1}$ can be one
of the points $p_{s_2}$ to $p_{s_5}$. 

The only points that can change the validity of the edge $p_ip_j$ over a sufficiently short time interval are the points $p_{s_i}$, $1\leq i\leq 5$. Therefore, we define at most four \textit{NotInWedge} certificates for the hull edge $p_ip_j$, certifying that the points $p_{s_i}$, $1\leq i\leq 5$, are outside the $k$-wedge (see Figure~\ref{fig:EDT_O}(a)). If $p_ip_j$ is adjacent to four other hull edges, this edge cannot be incident to a triangle, and if it is incident to a triangle, it cannot be adjacent to more
than two other hull edges. Let $t$ be the time when three points $p_i$, $p_j$, and $p_{s_i}$ are on the boundary of the $k$-wedge; at time $t^-$, $p_{s_i}$ is outside the
$k$-wedge. The hull edge $p_ip_j$ divides its corresponding
$k$-wedge $\overleftrightarrow{a_koa_{k+1}}$ into a bounded area
$\overline{op_ip_j}$ and an unbounded area
$\overleftrightarrow{a_kp_ip_ja_{k+1}}$ (see
Figure~\ref{fig:k-wedge}(b)). If $p_{s_i}$ moves inside the bounded
area $\overline{op_ip_j}$ at time $t^+$, the NotInWedge
certificate of $p_ip_j$ fails, and we must delete $p_ip_j$ from the
hull edges at time $t$ and replace it with two edges incident to
$p_{s_i}$. In Figure~\ref{fig:EDT_O}(a), if $p_{s_1}$ moves inside
the bounded area $\overline{op_ip_j}$, then we replace the hull 
edge $p_ip_j$ with two edges $p_ip_{s_1}, p_{s_1}p_j$; in
particular, the chain $[..., p_{s_2}p_i, p_ip_j, p_jp_{s_3},...]$ of
hull edges changes to $[..., p_{s_2}p_i, p_ip_{s_1}, p_{s_1}p_j,
p_jp_{s_3},...]$ when $p_{s_1}$ moves inside the $k$-wedge (see
Figure~\ref{fig:EDT_O}(b)). When this event occurs the previous
interior edges $p_ip_{s_1}$ and $p_{s_1}p_j$ become hull edges, and
we must replace the previous certificates of these edges with new
valid ones. If $p_{s_i}$  moves inside the unbounded area
$\overleftrightarrow{a_kp_ip_ja_{k+1}}$, without loss of generality let $p_{s_i}$ be
incident to $p_i$, we replace the hull edges $p_{s_i}p_i$ and
$p_ip_j$ with $p_{s_i}p_j$. Then the previous hull edge $p_ip_j$
either is an edge of $EDT_0$, in which case we must define a valid
certificate for it, or it is not, in which case we must delete it from
$EDT_0$ and add a new edge $p_{s_i}p_{s_1}$, where $p_ip_j$ is
incident to a triangle $p_ip_jp_{s_1}$; see Figure~\ref{fig:BadEvents}.
(a, b, and c).

\paragraph{\textbf{Consecutive Changes to EDT$_0$.}}\label{par:CC}
In some cases, when a certificate fails, we must apply a
\textit{sequence} of changes to $EDT_0$. These kinds of
changes occur at incident triangles, and as we will see, they can be handled
consecutively.

When a  NotInWedge certificate fails, we apply a sequence of edge insertions and edge deletions  to $EDT_0$. In Figure~\ref{fig:BadEvents}(a), when $p_{s_2}$ moves
inside the $k$-wedge of $p_ip_j$, we replace chain $p_{s_2}p_i,
p_ip_j$ of hull edges with $p_{s_2}p_j$ (see
Figure~\ref{fig:BadEvents}(b)), and then we apply a sequence of
changes;  the previous hull edge $p_ip_j$ is no longer an edge in
${\cal E}(EDT_0)$, because now the interior of its corresponding $0$-tri contains
the point $p_{s_2}$, and so we replace it with the edge
$p_{s_1}p_{s_2}$ (see Figure~\ref{fig:BadEvents}(c)). Finally, by
checking the $0$-tri's of other incident triangles, we can obtain a set
of valid edges for $EDT_0$ (see Figure~\ref{fig:BadEvents}(d)).

\begin{figure}[t]
\centering \includegraphics[scale=1]{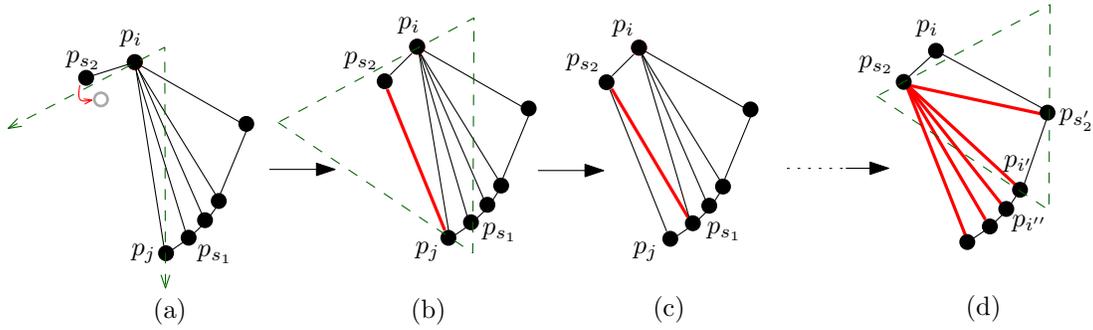} \caption{The
consecutive  changes to $EDT_0$ when $p_{s_2}$ moves inside the
$k$-wedge of $p_ip_j$.} \label{fig:BadEvents}
\end{figure}

A similar scenario could happen when a NotInTri certificate fails.
In Figure~\ref{fig:BadEvents}(d), if $p_i$ moves inside the $0$-tri
of $p_{s_2}$, $p_{s'_2}$, and $p_{i'}$, we must apply a
sequence of changes to $EDT_0$ that is the reverse of what we did above when the NotInWedge certificate failed. First we replace $p_{s_2}p_{s'_2}$ with
$p_ip_{i'}$. Then we must replace $p_{s_2}p_{i'}$ with
$p_ip_{i''}$, because $p_i$ is inside the $0$-tri of $p_{s_2}$,
$p_{i'}$, and $p_{i''}$. By checking the $0$-tri's of other incident
triangles we can obtain a valid set of edges for $EDT_0$; see
Figure~\ref{fig:BadEvents}, read from $(d)$ to $(a)$. Therefore, after any change to
$EDT_0$ we must check the validity of the incident triangles, which
 can be done easily.

Theorem~\ref{the:EDT_Changes} below enumerates the changes to the
Equilateral Delaunay graph (\ie, the Semi-Yao graph) when the points are moving and gives the
time to process all these events.

\begin{theorem}\label{the:EDT_Changes}
The number of changes to the Equilateral Delaunay graph, when the points move according to
polynomial functions of at most constant degree $s$, is
$O(n^2\beta_{s+2}(n))$. The total processing time for all events is
$O(n^2\beta_{s+2}(n)\log n)$.
\end{theorem}
\begin{proof}
From Lemma~\ref{the:SY_EDT}, the Equilateral Delaunay graph changes if and only if the Semi-Yao graph changes. Fix a  point $p_i$ and one of its wedges $W_l(p_i)$. Since the trajectory of each point $p_i(t)=(x_i(t),y_i(t))$ is defined by two polynomial functions of at most constant degree $s$, each point can insert into ${\cal V}_l(p_i)$ at most $s$ times. The $b_l$-coordinates of the points inserted into ${\cal V}_l(p_i)$ create at most $sn$ partial functions of at most constant degree $s$. From Theorem~\ref{the:partiallyDFcomplexity}, the minimum value of these $sn$ partial functions changes at most $\lambda_{s+2}(sn)$ times, which is equal to the number of all changes for the point with minimum $b_l$-coordinate among the points in ${\cal V}_l(p_i)$. Since $s$ is a constant, we have that $\lambda_{s+2}(sn)=O(\lambda_{s+2}(n))$. Thus the number of all changes for all points is $O(n\lambda_{s+2}(n))=O(n^2\beta_{s+2}(n))$.

The number of certificates is in the order of the number of changes to $EDT_0$. When a change to $EDT_0$ occurs, we update the $EDT_0$ and replace the invalid certificate(s) with new valid one(s). The time to make a constant number of deletions/insertions into the priority queue is $O(\log n)$. 

Thus the total time to process all events is $O(n^2\beta_{s+2}(n)\log n)$.
\end{proof}
\subsection{Kinetic All Nearest Neighbors}\label{sec:kinetic_ANN}
The Equilateral Delaunay graph (Semi-Yao graph) is a supergraph of the nearest neighbor graph. Let $Inc(p_i)$ be the set all edges incident to $p_i$ in the Semi-Yao graph. Over time, to maintain the nearest neighbor to each point $p_i$, we need to track the edge with the minimum length in $Inc(p_i)$.

Using a dynamic and kinetic tournament tree (see Section~\ref{sec:preliminary}), we can maintain the edge with the minimum length among the edges in $Inc(p_i)$. For each $Inc(p_i)$, $i=1,2,...,n$, we construct a dynamic and kinetic tournament tree ${\cal T}_i$.  The edges of $Inc(p_i)$ are stored at leaves of the tournament tree, and each of the internal nodes of the tree maintains the edge with  the minimum length stored at its two children; the root of the tree maintains the edge with minimum length among all edges in $Inc(p_i)$. 

Let $n_i$ be the cardinality of the set $Inc(p_i)$.  Consider a sequence of $m_i$ insertions and deletions into ${\cal T}_i$. From Theorem~\ref{the:DynamicKineticTT}, and the fact that the lengths of any two edges in $Inc(p_i)$ can become equal at most $2s$ times, the following results.

\begin{lemma}\label{the:DKTT}
The dynamic and kinetic tournament tree ${\cal T}_i$ of $n_i$ elements can be constructed in $O(n_i)$ time. The tournament tree
${\cal T}_i$ generates at most  $O(m_i\beta_{2s+2}(n_i)\log n_i)$ events, for a total cost of $O(m_i\beta_{2s+2}(n_i)\log^2 n_i)$.
\end{lemma}

Now we can prove the following.

\begin{corollary}\label{the:AllDKTT}
All the dynamic and kinetic tournament trees ${\cal T}_i$'s can be constructed in $O(n)$ time. These dynamic and kinetic tournament
trees generate at most $O(n^2\beta^2_{2s+2}(n)\log n)$ events, for a total cost of $O(n^2\beta^2_{2s+2}(n)\log^2 n)$.
\end{corollary}
\begin{proof}
By Lemma~\ref{the:DKTT} all the dynamic and kinetic tournament trees ${\cal T}_i$, $i=1,...,n$, generate at most $O(\sum_{i=1}^{i=n}m_i\beta_{2s+2}(n_i)\log n_i)=O(\beta_{2s+2}(n)\log n\sum_{i=1}^{i=n}m_i)$ events. Since each edge is incident to two points, inserting (resp. deleting) an edge $p_ip_j$ into the Equilateral Delaunay graph causes two insertions (resp. deletions) into the tournament trees ${\cal T}_i$ and ${\cal T}_j$. Therefore, by Theorem~\ref{the:EDT_Changes}, the number of all insertions/deletions into the tournament trees is $\sum_{i=1}^{i=n}m_i=O(n^2\beta_{s+2}(n))=O(n^2\beta_{2s+2}(n))$. Hence, the number of all events is $O(n^2\beta^2_{2s+2}(n)\log n)$, and the total cost is $O(n^2\beta^2_{2s+2}(n)\log^2 n)$. 
\end{proof}

Now we can prove the following theorem, which gives the results about our kinetic data structure for the all nearest neighbors problem.

\begin{theorem}\label{the:KinecitNNG}
Our kinetic data structure for maintenance of all the nearest neighbors uses linear space and $O(n\log n)$ preprocessing time. It handles $O(n^2\beta^2_{2s+2}(n)\log n)$ events with total processing time $O(n^2\beta^2_{2s+2}(n)\log^2 n)$. It is compact, efficient, responsive in an amortized sense, and local on average.
\end{theorem}
\begin{proof}
Since $\sum_i n_i = n$, the total size of all the tournament trees ${\cal T}_i$, $i=1,...,n$, is $O(n)$. The number of all edges in the EDG is $O(n)$. For each edge in the EDG, we define a constant number of certificates. Furthermore, the number of all certificates corresponding to the internal nodes of all ${\cal T}_i$ is linear. Thus the KDS is compact.  The ratio of the number of internal events $O(n^2\beta^2_{2s+2}(n)\log n)$ to the number of external events $O(n^2\beta_{2s})$ is polylogarithmic, which implies that the KDS is efficient. By Corollary~\ref{the:AllDKTT}, the ratio of the total processing time to the number of internal events is polylogarithmic, and so the KDS is responsive in an amortized sense. Since the number of all certificates is $O(n)$, each point participates in a constant number of certificates  on average, which implies that the KDS is local on average.
\end{proof}

\subsection{Kinetic Closest Pair}\label{sec:kinetic_CP}
The edge $p_ip_j$ with minimum length in the nearest neighbor graph gives the closest pair $(p_i,p_j)$. Since the Semi-Yao graph (EDG) is a supergraph of the nearest neighbor graph, to maintain the closest pair $(p_i,p_j)$ we need to maintain the edge with minimum length in the Semi-Yao graph. By constructing a dynamic and kinetic tournament tree, where the edges of the Semi-Yao graph are stored at the leaves of the dynamic and kinetic tournament tree, we can maintain the closest pair $(p_i,p_j)$ over time; the edge at the root of the dynamic and kinetic tournament tree gives the closest pair. The insertions and deletions into the dynamic and kinetic tournament tree occur when a change to the Semi-Yao graph occurs. Therefore, we can obtain the same results for maintenance of the closest pair over time as we obtained for maintenance of  all the nearest neighbors in Theorem~\ref{the:KinecitNNG}:

\begin{theorem}\label{the:KinecitCP}
Our kinetic data structure for maintenance of the closest pair uses linear space and $O(n\log n)$ preprocessing time. It handles
$O(n^2\beta^2_{2s+2}(n)\log n)$ events with total processing time $O(n^2\beta^2_{2s+2}(n)\log^2 n)$, and it is compact, efficient,
responsive in an amortized sense, and local on average.
\end{theorem}

\section{Yao Graph and EMST}\label{sec:YG_EMST}
Our approach for computing the Yao graph and the EMST is similar to the approach for computing all the nearest neighbors and the closest pair in Section~\ref{sec:ANN_CP_construction}. 

First we introduce a new supergraph  of the Yao graph, namely the Pie Delaunay graph, then we show how to maintain the Pie Delaunay graph (PDG) over time, and finally, using the kinetic version of the Pie Delaunay graph, we provide a KDS for maintenance of the Yao graph and the EMST when the points are moving.
\subsection{New Method for Computing the Yao Graph and the EMST}\label{sec:YG_EMST_construction}
Consider a partition of a unit disk into six \textit{pieces of pie} $\sigma_0,...,\sigma_5$, each of angle $\pi/3$ with common apex at the origin $o$. For $0\leq l \leq 5$, let $\sigma_l$ span the angular range $[(2l-1)\pi/6, (2l+1)\pi/6)$, and call any translated and scaled copy of $\sigma_l$ an \textit{$l$-pie}; see Figure~\ref{fig:disk}.

\begin{figure}[h]
  \begin{center}
    \includegraphics[scale=1]{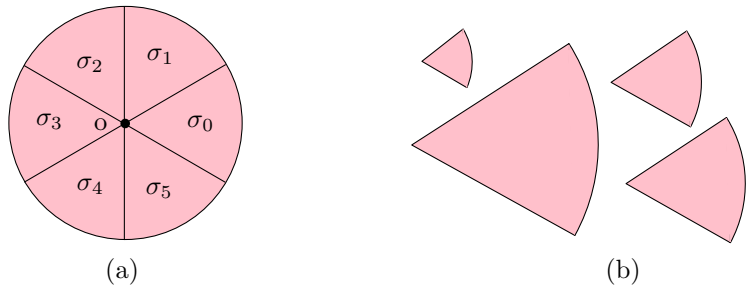}
  \end{center}
  \caption{(a) Partitioning the unit disk into six pieces of pie. (b) Some $0$-pie's.}
  \label{fig:disk}
\end{figure}

We define a Delaunay triangulation, which we call a \textit{Pie Delaunay triangulation}, of the set $P$ of $n$ points, based on the convex shape $\sigma_l$. Denote by $PDT_l$ the Pie Delaunay triangulation based on the $l$-pie. For two points $p_i$ and $p_j$ in $P$, the edge $p_ip_j$ is an edge of $PDT_l$ if and only if there is an empty $l$-pie such that $p_i$ and $p_j$ are on its boundary. We define the \textit{Pie Delaunay graph} (PDG) to be the union of all $PDT_l$ for $i=0,...,5$; \ie, $p_ip_j$ is a PDG edge if and only if it is an edge in $PDT_l$, where $0\leq l\leq 5$. 

The next lemma follows from Theorem~\ref{the:VD_DT_Con}.

\begin{lemma}\label{the:PDG_ConstructionTime}
The Pie Delaunay graph (PDG) can be constructed in $O(n\log n)$ time.
\end{lemma}

For each point $p_i\in P$, partition the plane into six wedges $W_0(p),..., W_5(p)$ of angle $\pi/3$ where $p_i$ is the common apex of the wedges. For $0\leq l\leq 5$, let $W_l(p_i)$ span the angular range $[(2l-1)\pi/6, (2l+1)\pi/6)$ around $p_i$. The \textit{Yao graph} can be constructed by connecting the point $p_i$ to its nearest points inside the wedges $W_l(p)$ for all $i=0,...,5$. We denote the Yao graph of a set of $n$ points by YG, the set of its edges by ${\cal{E}}(YG)$, and the set of Pie Delaunay graph edges by ${\cal E}(PDG)$. The following lemma shows that the Pie Delaunay graph is a supergraph of the Yao graph (YG).

\begin{lemma}\label{the:YG_SS_PDG}
${\cal E}(YG)\subseteq {\cal E}(PDG)$.
\end{lemma}
\begin{proof}
Assume edge $p_ip_j\in {\cal E}(YG)$ and let $p_j$ to be the nearest point to $p_i$ inside the wedge $W_l(p_i)$; see Figure~\ref{fig:yao}. The two sides of the wedge $W_l(p_i)$ are parallel to the two corresponding sides of $\sigma_l$, so there is an empty $l$-pie such that $p_i$ and $p_j$ lie on its boundary. Therefore, $p_ip_j\in PDT_l$ and hence it is an edge of the Pie Delaunay graph.
\end{proof}

\begin{figure}[t]
\centering
  \includegraphics[scale=1.3]{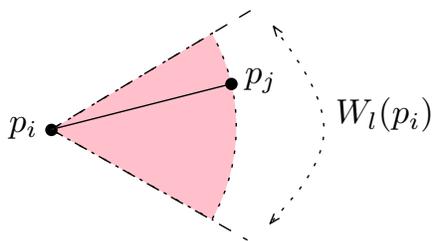}
  \caption{Nearest point to $p_i$ inside the wedge $W_l(p_i)$.}
  \label{fig:yao}
\end{figure}

Now we can state and prove the main result of this section.

\begin{theorem}\label{the:YG_EMST_ConstructionTime}
The Yao graph and the EMST can be constructed in $O(n\log n)$ time.
\end{theorem}
\begin{proof}
The Pie Delaunay graph is the union of six Pie Delaunay triangulations, which implies that it has a linear number of edges. By Lemma~\ref{the:YG_SS_PDG}, the Pie Delaunay graph is a supergraph of the Yao graph. Thus by tracing over the edges incident to each point $p_i$, we can find the edge with minimum length inside each wedge $W_l(p_i)$, for $l=0,...,5$; this gives the Yao graph. Since the Pie Delaunay graph can be constructed in time $O(n\log n)$ (by Lemma~\ref{the:PDG_ConstructionTime}), the Yao graph can be constructed in time $O(n\log n)$.

The Yao graph is a supergraph of the EMST~\cite{DBLP:journals/siamcomp/Yao82}. Thus the minimum spanning tree of the Yao graph is equal to the EMST. Since the cardinality of the set of edges in the Yao graph graph is at most $6n$, the EMST can be constructed using the Prim algorithm~\cite{Prim57} or the Kruskal algorithm~\cite{Kruskal1956} in time $O(n\log n)$.
\end{proof}
\subsection{Kinetic Pie Delaunay Graph}\label{sec:kinetic_PDG}
Our KDS for maintenance of the Pie Delaunay graph is similar to the KDS for maintenance of the Equilateral Delaunay graph in Section~\ref{sec:kinetic_EDG}. The Pie Delaunay graph (PDG) is the union of all $PDT_l$, for $l=0,..,5$:  ${\cal E}(PDG)=\bigcup_l {\cal E}(PDT_l)$. Here, we only provide a KDS for $PDT_0$; the other $PDT_l$, for $l=1,..,5$, are handled similarly.

Similar to Section~\ref{sec:kinetic_EDG}, we call each edge that is not on the boundary of the infinite face of $PDT_0$ an \textit{interior edge} and the other edges on the boundary of the infinite face \textit{hull edges}, and corresponding to them we define two kinds of certificates, \textit{NotInCone} and \textit{NotInPie}, respectively.
\paragraph{Interior Edges.} 
By definition, an interior edge $p_{i'}p_{j'}\in {\cal E}(PDT_0)$ is incident to two triangles of $PDT_0$ that together form a \textit{quadrilateral}. Let $p_{r'}$ and $p_r$ be the two other vertices of the quadrilateral. For the edge $p_{i'}p_{j'}$, we define a \textit{NotInPie} certificate which certifies that point $p_{r}$ (resp. $p_{r'}$) is outside the $0$-pie passing through $p_{i'}$, $p_{j'}$, and $p_{r'}$ (resp. $p_r$). When the certificate fails, we replace $p_{i'}p_{j'}$ by $p_{r}p_{r'}$. In general, when the certificates corresponding to an interior edge fails, we perform such an edge swap.
\paragraph{Hull Edges.}
Let $o$, $w_0$, and $w_1$ be vertices of a $0$-pie (see Figure~\ref{fig:NotInCone}(a)). Two of the edges on the boundary of the $0$-pie are line segments and one of them is an arc; denote the line segments by $\overline{ow_0}$ and $\overline{ow_1}$ and the arc by $\overline{w_0w_1}$. By removing one of them and extending the line segment(s) to infinity, a cone can be created. We call these cones \textit{$k$-cones}. By definition, the edge $p_ip_j$ is a hull edge of $PDT_0$ if and only if there exists an empty $k$-cone such that $p_i$ and $p_j$ are on its boundary. 

Consider the $k$-cone $ow_1w_0$ corresponding to the edge $p_ip_j$ where one of the endpoints $p_i$ lies on the half-line $\overrightarrow{w_0o}$ and the other point $p_j$ lies on the half-arc $\overrightarrow{w_0w_1}$ (see Figure~\ref{fig:NotInCone}(b)). Let $\overrightarrow{\tilde{w_1}\tilde{w_0}}$ be the half-line perpendicular to  $\overrightarrow{w_1o}$ through $p_j$. For such a $k$-cone we assume that the line segment $\overrightarrow{w_1o}$ goes to infinity. This means that $w_1$ (resp. $w_0$) tends to $\tilde{w_1}$ (resp. $\tilde{w_0}$) and the $k$-cone approaches a right-angled wedge; see Figure~\ref{fig:NotInCone}(c).


Each hull edge $p_ip_j$ is adjacent to at most four other hull edges, denoted by $p_ip_{s_2}$, $p_ip_{s_3}$, $p_jp_{s_4}$, $p_jp_{s_5}$, and incident to at most one triangle. Let $p_{s_1}$ be the third vertex of this triangle if it exists; $p_{s_1}$ can be one of the $s_i$ where $2\leq i\leq 5$. If $p_ip_j$ is adjacent to at most four other triangles, then it cannot be incident to a triangle. In particular, at any time, the number of points $p_{s_i}$ is at most four. Therefore, for the $k$-cone passing through $p_i$ and $p_j$, we define at most four \textit{NotInCone} certificates certifying that the $p_{s_i}$ are outside of the $k$-cone. Note that in the case that a $k$-cone approaches a right-angled wedge (see Figure~\ref{fig:NotInCone}(c)), the certificate of the hull edge $p_ip_j$ fails when a point either crosses the half-line $\overrightarrow{w_1o}$, or reaches the line-segment $\overline{\tilde{w_1}p_j}$, or crosses the half-line $\overrightarrow{p_j\tilde{w_0}}$.

\begin{figure}[t]
  \centering
  \includegraphics[scale=1.2]{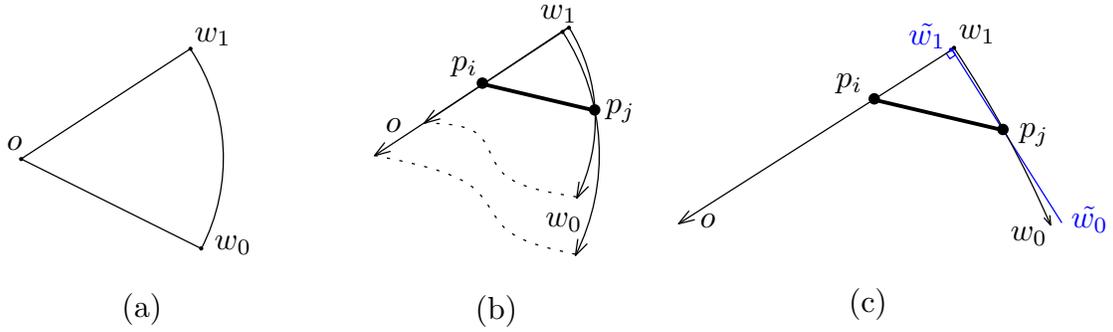}
  \caption{(a) A $0$-pie. (b) Two $k$-cones corresponding to the hull edge $p_ip_j$. (c) The $k$-cone approaches a right-angled wedge as $o$ goes to infinity.}
  \label{fig:NotInCone}
\end{figure}

The changes that can occur to $PDT_0$ are similar to the changes to $EDT_0$ and can easily be handled; see the paragraph "Consecutive Changes to EDT$_0$" in Section~\ref{sec:kinetic_EDG} for more details.


Next we state a theorem that enumerates the number of the combinatorial changes to the Pie Delaunay graph.

\begin{theorem}\label{the:num_changes_PDG}
The number of all changes (edge insertions and edge deletions) to the Pie Delaunay graph of a set of $n$ moving points with trajectories given by polynomial functions of at most constant degree $s$ is $O(n^3\beta_{2s+2}(n))$.
\end{theorem} 
\begin{proof}
Consider $PDT_0$. The number of hull-edge changes to $PDT_0$ is $O(n^3)$ as three points are involved in any hull change.  Since $n^3 =O(n^3\beta_{2s+2}(n))$, we focus on the number of changes to the triangles of $PDT_0$. 

For each edge $p_ip_j$ of a triangle in $PDT_0$, four different cases are possible as shown in Figure~\ref{fig:com_changes}. It is easy to see for any triangle $\Delta$ in the $PDT_0$ that case (a) of Figure~\ref{fig:com_changes} may happen to one of its edges. We charge any change to $\Delta$ to this edge. Therefore, we consider the number of combinatorial changes to $PDT_0$ for an arbitrary edge $p_ip_j$ that satisfies case (a) of Figure~\ref{fig:com_changes}.

As mentioned above, two edges of a $0$-pie are line segments $\overline{ow_0}$ and $\overline{ow_1}$ and one of them is an arc $\overline{w_0w_1}$. Let $C_{w_0w_1}$ be the cone whose sides are created by removing the arc $\overline{w_0w_1}$ of the $0$-pie and extending the two line segments to infinity; the wedge $C_{w_0w_1}$ is the area between two half-lines $\overrightarrow{ow_0}$ and $\overrightarrow{ow_1}$. Let ${\cal V}(C_{w_0w_1})$ be the set of all points inside the wedge $C_{w_0w_1}$. In Figure~\ref{fig:com_changes}(a), a change for triangle $p_ip_jp_r$ corresponding to $p_ip_j$ occurs in two cases:
\\
\textit{Case (I).} For some $p_t\in {\cal V}(C_{w_0w_1})$, the length of the edge $op_t$ becomes smaller than the length of the edge $op_r$.

Note that since the degree of each function describing each point's motion is at most $s$, each point of $P$ except $p_i$ and $p_j$ can move inside the cone $C_{w_0w_1}$ at most $s$ times. Summing over all points in $P$ there are $O(sn)$ insertions into ${\cal V}(C_{w_0w_1})$. The distance of these points from the apex $o$, in the $L_2$ metric, creates $O(sn)$ partial functions, and each pair of these functions intersects at most $2s$ times. Therefore, the number of combinatorial changes corresponding to an arbitrary edge $p_ip_j$ equals $\lambda_{2s+2}(sn)$, which is equal to the number of breakpoints in the lower envelope of $sn$ partial functions of at most degree $2s$ (see Theorem~\ref{the:partiallyDFcomplexity}). Since the maximum degree $s$ is a constant, $\lambda_{2s+2}(sn)=O(\lambda_{2s+2}(n))$. The number of all possible edges is $O(n^2)$, and therefore the number of combinatorial changes corresponding to all edges is $O(n^2\lambda_{2s+2}(n))$.
\\
\textit{Case (II).} In addition to the above changes for the edge $p_ip_j$ in Case (I), there exist other changes that can occur when a point such as $p_{t'}$ passes through the segment $op_i$ or the segment $op_j$ and enters inside the area $op_ip_j$ (see Figure~\ref{fig:com_changes}(a)). Map each point
$p_i=(x_i(t),y_i(t))$ to a point $p'_i=(u_i(t),v_i(t))$ in a new parametric plane where $u_i(t)=x_i(t)+\sqrt{3}y_i(t)$ and $v_i(t)=x_i(t)-\sqrt{3}y_i(t)$. Passing the point $p_{t'}$ through  the segment $op_i$ or the segment $op_j$ means that the point $p_{t'}$ exchanges its $u$-coordinate or its $v$-coordinate with the $u$-coordinate or $v$-coordinate of $p'_i$ or $p'_j$. We call these changes
\textit{swap-changes}. Observe that the total number of swap-changes for all cases is bounded by the number of all swaps between points in their ordering with respect to the $u$-axis and $v$-axis. The number of all the $u$-swaps and $v$-swaps between points is at most $O(n^2)$. 

Hence, the number of changes to the Pie Delaunay graph is $O(n^3\beta_{2s+2}(n))$.
\end{proof}

\begin{figure}[t]
  \centering
  \includegraphics[scale=0.9]{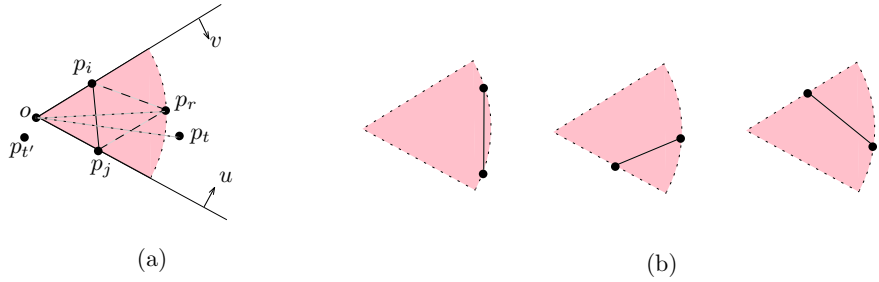}
  \caption{Combinatorial changes for an arbitrary edge $p_ip_j$.}
  \label{fig:com_changes}
\end{figure}

After any change to the Pie Delaunay graph, we replace a constant number of  (invalid) certificates from the priority queue with new valid ones, which takes $O(\log n)$ time. From the above discussion, together with Lemma~\ref{the:PDG_ConstructionTime} and Theorem~\ref{the:num_changes_PDG}, we obtain the following theorem.

\begin{theorem}\label{the:pie_d}
For a set of $n$ points in the plane with trajectories given by polynomial functions of at most constant degree $s$, there exists a KDS for maintenance of the Pie Delaunay graph that uses linear space, $O(n\log n)$ preprocessing time, and  that processes $O(n^3\beta_{2s+2}(n))$ events with total processing time $O(n^3\beta_{2s+2}(n)\log n)$.
\end{theorem}
\subsection{Kinetic Yao Graph}\label{sec:kinetic:YG}
To maintain the Yao graph, for each point $p_i\in P$, we must maintain the nearest points to $p_i$ inside the wedges $W_l(p_i)$, where $0\leq l\leq 5$. Since the Yao graph is a subgraph of the Pie Delaunay graph (by Lemma~\ref{the:SY_EDT}), to maintain the nearest points inside the wedges of $p_i$, we only need to track the edges of the Pie Delaunay graph incident to $p_i$ with minimum length inside the wedges $W_l(p_i)$ for all $l=0,...,5$.

Let $Inc_l(p_i)$ be the set all edges of the Pie Delaunay graph incident to $p_i$ inside the wedge $W_l$. We store the edges of $Inc_l(p_i)$ at leaves of a dynamic and kinetic tournament tree	${\cal T}_{l,i}$ (see Section~\ref{sec:preliminary}). The root of ${\cal T}_{l,i}$ maintains the winner, the edge with minimum length among all edges in $Inc_l(p_i)$. Given the KDS of the Pie Delaunay graph and making an analysis similar to that of Corollary~\ref{the:AllDKTT} and Theorem~\ref{the:KinecitNNG}, the following theorem results.

\begin{theorem}\label{the:kineticYG}
The KDS for maintenance of the Yao graph uses $O(n)$ space, $O(n\log n)$ preprocessing time, and processes $O(n^3\beta^2_{2s+2}\log n)$ (internal) events with total processing time $O(n^3\beta^2_{2s+2}\log^2 n)$. It is compact, responsive in an amortized sense, and local on average, but it is not efficient.
\end{theorem}

For \textit{linearly} moving points in the plane, Katoh~\etal~\cite{Katoh:1992:MMS:1398516.1398835} showed  that the number of changes to the Yao graph is $O(n\lambda_4(n))$. In the following theorem we bound the number of combinatorial changes to the Yao graph of a set of moving points whose trajectories are given by polynomial functions of at most constant degree $s$. For maintenance of the Yao graph, our KDS processes $O(n^3\beta^2_{2s+2}\log n)$ events, but the following theorem proves that the number of exact changes to the Yao graph is nearly quadratic, which explains why our KDS is not efficient. 

\begin{theorem}\label{the:num_changes_YG}
The number of all changes to the Yao graph, when the points move with polynomial trajectories of at most constant degree $s$,  is $O(n^2\beta_{2s+2}(n))$.
\end{theorem}
\begin{proof}
Consider the point $p_i\in P$ and one of its wedges $W_l(p_i)$. Each of the other points in $P$ can be moved inside the wedge $W_l(p_i)$ at most $s$ times, and so there exist $O(sn)$ insertions into the wedge $W_l(p_i)$. The distance of these points from $p_i$ creates $O(sn)$ partial functions; each pair of these functions intersects at most $2s$ times. By Theorem~\ref{the:partiallyDFcomplexity}, the edge with minimum length changes at most $\lambda_{2s+2}(sn)=O(\lambda_{2s+2}(n))$ times.
 
Hence, the number of all changes to the Yao graph of a set of $n$ moving points is $O(n\lambda_{2s+2}(n))$.
\end{proof}

\begin{remark}
Using an argument similar to that for the KDS we obtained for the Yao graph in the $L_2$ metric, a KDS for the Yao graph in the $L_1$ and $L_\infty$ metrics can be obtained. 

Denote by $\square$ the \textit{unit square} with corners at $(0,0)$, $(1, 0)$, $(0, 1)$, and $(1, 1)$ in a Cartesian coordinate system, and call any translated and scaled copy of $\square$ an \textit{SQR}. The edge $p_ip_j$ is an edge of the Delaunay triangulation based on an SQR in the $L_\infty$ metric if and only if there is an empty SQR such that $p_i$ and $p_j$ are on its boundary (\ie, the interior of SQR contains no point of $P$). Abam~\etal~\cite{Abam:2010:SEK:1630166.1630284} showed how to maintain a Delaunay triangulation based on a diamond. Each SQR is a diamond, so using their approach applies. The Delaunay triangulation where the triangulation is based on an SQR in the $L_\infty$ metric can be maintained kinetically by processing at most $O(n\lambda_{s+2}(n))$ events, each in amortized time $O(\log n)$. The Delaunay triangulation based on an SQR is a supergraph for the Yao graph in the $L_\infty$ metric. Therefore, we can have a KDS for the Yao graph in the $L_\infty$ metric that uses $O(n)$ space, $O(n\log n)$ preprocessing time, and that processes $O(n^2\beta^2_{s+2}(n)\log n)$ events, each in amortized time $O(\log n)$.

The Delaunay triangulation in the $L_1$ metric can be constructed/maintained analogously, by rotating all points $45$ degrees around the origin and constructing/maintaining the Delaunay triangulation in the $L_\infty$ metric.
\end{remark}
\subsection{Kinetic EMST}\label{sec:kinetic_EMST}
Our kinetic approach for maintaining the EMST is based on the fact that the EMST is a subgraph of the Yao graph, where the number of the wedges around each point in the Yao graph is greater than or equal to six~\cite{DBLP:journals/siamcomp/Yao82}. 

Let $L$ be a list of the Yao graph edges (which in fact are stored at the roots of the dynamic and kinetic tournament trees ${\cal T}_{l,i}$, for each point $p_i\in P$ and $l=1,...,6n$), sorted with respect to their Euclidean lengths. A change to the EMST may occur when two edges in $L$ change their ordering. For each two consecutive edges in  $L$, we define a certificate certifying the respective sorted order of the edges. Whenever the ordering of two edges in this list is changed, we apply the required changes to the EMST KDS. Therefore, to update the EMST when the points are moving, we must track the changes to $L$. There exist two types of changes to $L$: $(a)$ edge insertion and edge deletion from $L$, and $(b)$ a change in the order of two consecutive edges in $L$. The following discusses how to handle these two types of events.

\paragraph{Case (a):} As soon as an edge is deleted from $L$ a new one is inserted. Both the deleted edge and the inserted edge are in the same dynamic and kinetic tournament tree, and both of them have a common endpoint; see Figure~\ref{fig:InsDel}. Call the deleted edge and the inserted edge $p_ip_j$ and $p_ip_r$, respectively, and denote by ${\cal T}_{i,l}$ the dynamic and kinetic tournament tree that contains $p_ip_j$ and $p_ip_r$. The deleted edge $p_ip_j$ can be one of the EMST edges at time $t^-$ and if so, we have to find a new edge to repair the EMST at time $t^+$. The following lemma proves that this new edge is $p_ip_r$.

\begin{figure}[t!]
\centering
  \includegraphics[scale=0.9]{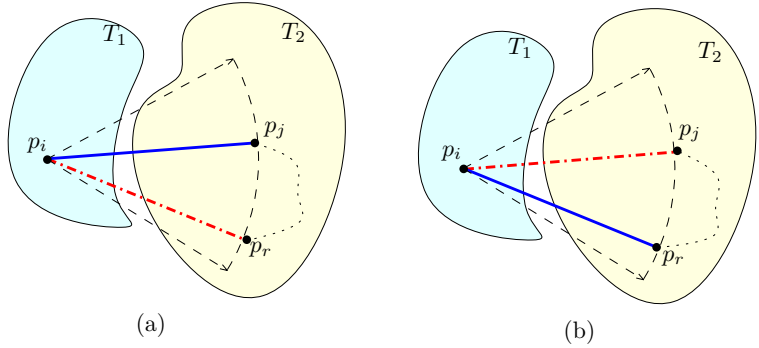}
  \caption{The edge connecting two subtrees $T_1(P_1,E_1)$ and $T_2(P_2,E_2)$: (a) At time $t^-$, $|p_ip_r|>|p_ip_j|>|p_jp_r|$ and the edge connecting $T_1$ and $T_2$ is $p_ip_j$. (b) At time $t^+$, $|p_ip_j|>|p_ip_r|>|p_jp_r|$ and the edge connecting $T_1$ and $T_2$ is $p_ip_r$.}
  \label{fig:InsDel}
\end{figure}

\begin{lemma}
Let $p_ip_j$ be the winner  of the dynamic and kinetic tournament tree ${\cal T}_{i,l}$. Suppose $p_ip_j\in {\cal E}(EMST)$ at time $t^-$ and let $p_ip_r$ be the winner of ${\cal T}_{i,l}$ at time $t^+$. Then ($i$) at time $t^-$, $p_ip_r\notin {\cal E}(EMST)$, and ($ii$) at time $t^+$, $p_ip_r\in {\cal E}(EMST)$ and $p_ip_j\notin {\cal E}(EMST)$.
\end{lemma}
\begin{proof}
Deleting an edge $p_ip_j$ from EMST creates two subtrees $T_1(P_1,E_1)$ and $T_2(P_2,E_2)$. Let $p_i\in P_1$ and $p_j\in P_2$; see Figure~\ref{fig:InsDel}.  At time $t^-$, since $p_ip_j\in {\cal E}(EMST)$, $|p_ip_r|>|p_ip_j|>|p_jp_r|$, and $\angle p_jp_ip_r \leq \pi/3$, we have that $p_r\in P_2$. This can be concluded by contradiction. Thus ($i$) at time $t^-$, $p_ip_r\notin {\cal E}(EMST)$. 

The proof that $p_ip_j\notin {\cal E}(EMST)$ at time $t^+$ is analogous to the proof for ($i$). Therefore, at time $t^+$, the EMST is the union of two trees $T_1$ and $T_2$ and the edge $p_ip_r$.
\end{proof}

\paragraph{Case (b):} Let  $path(e)$ be the simple path in the EMST between the endpoints of edge $e$ and let $|e|$ be the Euclidean length of $e$. A change in the sorted list $L$ corresponds to a pair of edges $e$ and $e'$ in ${\cal E}(YG)$ such that at time $t^-$, $|e|<|e'|$, and at time $t^+$, $|e|>|e'|$. Thus at time $t$, $e$ may be replaced by $e'$ in the EMST. It is easy to see the following.

\begin{observation}
The EMST changes if and only if at time $t^-$, $|e|<|e'|$, $e\in {\cal E}(EMST)$, $e'\notin {\cal E}(EMST)$, $e\in path(e')$, and at time $t^+$, $|e|>|e'|$.
\end{observation}

Such events can be detected and maintained in $O(\log n)$ time per operation using the \textit{link-cut tree data structure} of Sleator and Tarjan~\cite{Sleator:1983:DSD:61337.61338}.

Given a KDS for maintenance of the Yao graph, the following bounds the number of events for maintaining the EMST.

\begin{lemma}\label{the:num_changes_EMST}
Given a Yao graph KDS for a set of $n$ points moving with polynomial trajectories of constant maximum degree $s$, there exists a KDS for maintenance of the EMST that processes $O(n^3\beta_{2s+2}(n))$ events.
\end{lemma}
\begin{proof}
The set of Yao graph edges is a superset of the set of the EMST edges, and any change in the order of consecutive edges in the sorted list $L$ of the Yao graph edges may change the EMST. More precisely, any edge insertion/deletion in the Yao graph implies an insertion/deletion into $L$, and each insertion may cause $O(n)$ changes in the EMST. From Theorem~\ref{the:num_changes_YG}, the number of all insertions and deletions into the sorted list $L$ is $O(n^2\beta_{2s+2}(n))$. Thus the number of events that our KDS processes is $O(n^3\beta_{2s+2}(n))$.
\end{proof}

The KDS for maintenance of the EMST uses the Pie Delaunay graph KDS and the Yao graph KDS. From the above discussion and Theorems~\ref{the:pie_d} and \ref{the:kineticYG}, the following results.

\begin{theorem}\label{the:kineticEMST}
The KDS for maintenance of the EMST uses linear space and requires $O(n\log n)$ preprocessing time.  The KDS processes $O(n^3\beta^2_{2s+2}(n)\log n)$ events, each in amortized time $O(\log n)$. The KDS is compact, responsive in an amortized sense, and local on average.
\end{theorem}
\section{Discussion and Open Problems}\label{sec:conclusion}
We have provided a kinetic data structure for the all nearest neighbors problem for a set of moving points in the plane. We have applied our structure
to maintain the closest pair as the points move. Comparison of our algorithm with the algorithm of Agarwal~\etal~\cite{Agarwal:2008:KDD:1435375.1435379} shows that in $\mathbb{R}^2$, our deterministic algorithm is simpler and more efficient than their randomized algorithm for maintaining  all the nearest neighbors. In $\mathbb{R}^3$, the number of edges of the Equilateral Delaunay graph  is $O(n^2)$,  and so for maintenance of  all the nearest neighbors, our kinetic approach needs $O(n^2)$ space. By contrast, the randomized kinetic data structure by Agarwal~\etal~\cite{Agarwal:2008:KDD:1435375.1435379} uses $O(n\log^3n)$ space. Thus, for higher dimensions ($d\geq 3$), their approach is asymptotically more efficient, but the simplicity
of our algorithm may make it more attractive. In higher dimensions, our deterministic method of maintaining the Equilateral Delaunay graph, does not satisfy all four kinetic
performance criteria. Thus, finding a deterministic kinetic algorithm for maintenance of  all the nearest neighbors in higher dimensions, and that satisfies the performance criteria, is a future direction.

We have also provided a KDS for maintenance of the EMST and the Yao graph on a set of $n$ moving points. Our KDS for maintenance of the EMST processes $O(n^3\beta^2_{2s+2}(n)\log n)$ events, which improves the previous $O(n^4)$ bound of Rahmati and Zarei~\cite{DBLP:conf/iwoca/RahmatiZ11}. The kinetic algorithm of Rahmati and Zarei results in a KDS having $O(n^{3+\epsilon})$ events, for any $\epsilon >0$, under the assumptions that ($i$) any four points can be co-circular at most twice, and ($ii$) either no ordered triple of points can be collinear more than once, or no triple of points can be collinear more than twice. Our kinetic approach further improves the upper bound $O(n^{3+\epsilon})$ under the above assumptions. A tight upper bound is not known. Our KDS can also be used to maintain an $L_1$-MST and an $L_\infty$-MST. By defining the Pie Delaunay graph and the Yao graph in $\mathbb{R}^d$, our kinetic approach can be used to give a simple KDS for the EMST in higher dimensions, but this approach does not satisfy all the performance criteria.

For linearly moving points in the plane,  Katoh~\etal~\cite{Katoh:1992:MMS:1398516.1398835} proved an upper bound of $O(n^32^{\alpha(n)})$ (resp. $O(n^{5/2}\alpha(n)$) for the number of combinatorial changes of the EMST (resp. $L_1$-MST and $L_\infty$-MST), where $\alpha(n)$ is the inverse Ackermann function. The upper bound was later proved to $O(\lambda_{ps+2}(n)n^{2-1/(9.2^{ps-3}})\log^{2/3}n)$ for the $L_p$-MST in $\mathbb{R}^d$, where the coordinates of the points are polynomial functions of constant maximum degree $s$~\cite{chan__2003}; for $p=2$ and $s=1$, this formula gives the first improvement $O(n^{25/9}2^{\alpha(n)}\log^{2/3}n)$ over Katoh~\etal's $O(n^32^{\alpha(n)})$ bound. An even better bound $O(n^{8/3}2^{\alpha(n)}\log^{4/3}n)$ can be obtained by combining the results of Chan~\cite{chan__2003} with those of Marcus and Tardos~\cite{Marcus:2006:IRS:1142905.1142913}. Finding a tight upper bound for the number of combinatorial changes  of the EMST, and finding a KDS for the EMST in $\mathbb{R}^d$ that processes a sub-cubic number of events are other future directions. 
\bibliographystyle{splncs03}
\bibliography{References}
\end{document}